\documentclass[sigplan,10pt]{acmart}
\AtBeginDocument{%
  \providecommand\BibTeX{{%
    \normalfont B\kern-0.5em{\scshape i\kern-0.25em b}\kern-0.8em\TeX}}}

\setcopyright{acmcopyright}
\copyrightyear{2021}
\acmYear{2021}
\acmConference[SOSP '21]{ACM SIGOPS 28th Symposium on Operating Systems Principles}{October 26--29, 2021}{Virtual Event, Germany}
\acmBooktitle{ACM SIGOPS 28th Symposium on Operating Systems Principles (SOSP '21), October 26--29, 2021, Virtual Event, Germany}
\acmPrice{}
\acmDOI{10.1145/3477132.3483561}
\acmISBN{978-1-4503-8709-5/21/10}

\usepackage{xspace}
\usepackage{multirow}
\usepackage{comment}
\usepackage{xfrac}
\usepackage{fancybox, fancyvrb, calc}
\usepackage[center, tight]{subfigure}
\usepackage[inline]{enumitem}
\usepackage{amsmath}
\usepackage[font=small]{caption}
\usepackage{url}
\usepackage{hyperref}
\usepackage{graphicx}
\usepackage{pbox}
\usepackage{algorithm, algorithmicx}
\usepackage[noend]{algpseudocode}
\usepackage{balance}
\usepackage{wrapfig}
\usepackage{tikz, pgfplots}
\usetikzlibrary{patterns,matrix,positioning,shadows,backgrounds,arrows,calc,fit,automata,shapes.geometric,arrows,decorations.pathreplacing,decorations.markings,shadings,shapes.symbols,arrows.meta}
\usepackage[capitalise]{cleveref}
\usepackage{tabularx}
\usepackage[T1]{fontenc}
\usepackage[utf8]{inputenc}
\usepackage[normalem]{ulem}
\usepackage{wasysym}
\usepackage{enumitem}
\usepackage{sidecap}
\usepackage{caption}
\usepackage{pifont}
\usepackage[mathcal]{euscript}
\usepackage{setspace}
\usepackage{listings}
\usepackage{footnote}
\usepackage{times}
\makesavenoteenv{tabular}
\makesavenoteenv{table}
\usepackage{soul}
\usepackage[noindentafter]{titlesec}

\titlespacing\section{0pt}{5pt plus 3pt minus 2pt}{5pt plus 3pt minus 2pt}
\titlespacing\subsection{0pt}{5pt plus 3pt minus 2pt}{5pt plus 3pt minus 2pt}

\usepackage{hyperref}
\hypersetup{
  colorlinks=true,      % false: boxed links; true: colored links
  linkcolor=blue,       % color of internal links
  citecolor=magenta,    % color of links to bibliography
  filecolor=cyan,       % color of file links
  urlcolor=red          % color of external links
}
\definecolor{listinggray}{gray}{0.9}
\definecolor{lbcolor}{rgb}{0.9,0.9,0.9}
\graphicspath{{fig/}}

\usepackage{etoolbox}
\theoremstyle{plain}
\newtheorem{myTheorem}{Theorem}[section]

\definecolor{boxclr}{gray}{0.9}

\pgfdeclarelayer{bg}    % declare background layer
\pgfsetlayers{bg,main}  % set the order of the layers (main is the standard layer)

\makeatletter
\newcommand{\thickhline}{%
    \noalign {\ifnum 0=`}\fi \hrule height 0.8pt
    \futurelet \reserved@a \@xhline
}
\newcolumntype{"}{@{\vrule width 0.8pt}}
\newcolumntype{[}{@{\vrule width 0.8pt\hskip\tabcolsep}}
\newcolumntype{]}{@{\hskip\tabcolsep\vrule width 0.8pt}}
\newcolumntype{!}{@{\hskip\tabcolsep\vrule width 0.8pt\hskip\tabcolsep}}
\makeatother

\newcommand{\cppsnippet}[1]{%
  \begin{lstlisting}[gobble=4]
    #1
  \end{lstlisting}
}

\tikzset{
    position/.style args={#1:#2 from #3}{
        at=(#3.#1), anchor=#1+180, shift=(#1:#2)
    }
}

\tikzset{
  half fill/.style 2 args={fill=#2, path picture={
    \fill[#1, sharp corners] (path picture bounding box.west) --
                         (path picture bounding box.east) --
                         (path picture bounding box.south east) --
                         (path picture bounding box.south west) -- cycle;}},
}

\tikzset{
  nil fill/.style 2 args={fill=#2, path picture={
    \fill[#1, sharp corners] (path picture bounding box.155) --
                         (path picture bounding box.25) --
                         (path picture bounding box.south east) --
                         (path picture bounding box.south west) -- cycle;}},
}

\tikzset{
  almost fill/.style 2 args={fill=#2, path picture={
    \fill[#1, sharp corners] (path picture bounding box.205) --
                         (path picture bounding box.335) --
                         (path picture bounding box.south east) --
                         (path picture bounding box.south west) -- cycle;}},
}

\definecolor{oldcolor}{HTML}{c66541}

\newcommand*\circled[1]{\tikz[baseline=(char.base)]{
            \node[shape=circle,draw,inner sep=0.5pt] (char) {\small #1};}}

\theoremstyle{nonumberplain}
\algdef{SE}[DOWHILE]{Do}{doWhile}{\algorithmicdo}[1]{\algorithmicwhile\ #1}%

\definecolor{burntorange}{rgb}{0.8, 0.33, 0.0}

\newcommand{\cut}[1]{}

\newcommand{\paragraphb}[1]{\vspace{0.075in}\noindent{\bf #1.}}

\colorlet{soulgreen}{green!30}

\colorlet{soulred}{blue!20}
 
\colorlet{soulpeach}{red!20}

\cut{
\usepackage{draftwatermark}
\SetWatermarkText{DRAFT -- DO NOT REDISTRIBUTE \ \ \ \ \ \ \ DRAFT -- DO NOT REDISTRIBUTE}
\SetWatermarkFontSize{1.5cm}
\SetWatermarkAngle{307}
\SetWatermarkLightness{0.85}
}
\definecolor{darkgreen}{rgb}{0.1, 0.14, 0.13}

\tikzset{ 
table/.style={
  matrix of nodes,
  row sep=-\pgflinewidth,
  column sep=-\pgflinewidth,
  nodes={rectangle,thick,draw=black,text width={},align=center,font=\small},
  text depth=0.25ex,
  text height=1.25ex,
  nodes in empty cells
},
map/.style={
  matrix of nodes,
  row sep=-\pgflinewidth,
  column sep=-\pgflinewidth,
  nodes={rectangle,draw=black,text width=5em,align=center,font=\small},
  text depth=0.25ex,
  text height=1.25ex,
  nodes in empty cells
},
bigmap/.style={
  matrix of nodes,
  row sep=-\pgflinewidth,
  column sep=-\pgflinewidth,
  nodes={rectangle,draw=black,text width=26em,align=center,font=\small},
  text depth=0.25ex,
  text height=1.25ex,
  nodes in empty cells
},
memcell/.style={
  draw, 
  very thick, 
  text width=0.25em, 
  text height=0.25em
},
}

\tikzstyle{startstop} = [rectangle, rounded corners, minimum width=3em, minimum height=1em,text centered, draw=black, fill=red!30]
\tikzstyle{io} = [trapezium, trapezium left angle=70, trapezium right angle=120, minimum width=2.5em, minimum height=1em, text centered, draw=black, fill=blue!30]
\tikzstyle{process} = [rectangle, minimum width=1.5em, minimum height=1em, align=center, draw=black, fill=gray!30]
\tikzstyle{decision} = [diamond, minimum width=3em, minimum height=1em, align=center, draw=black, fill=SkyBlue!30]
\tikzstyle{arrow} = [thick,->,>=stealth]
\tikzstyle{monolog} = [fill=SkyBlue!30]

\tikzset{
  on each segment/.style={
    decorate,
    decoration={
      show path construction,
      moveto code={},
      lineto code={
        \path [#1]
        (\tikzinputsegmentfirst) -- (\tikzinputsegmentlast);
      },
      curveto code={
        \path [#1] (\tikzinputsegmentfirst)
        .. controls
        (\tikzinputsegmentsupporta) and (\tikzinputsegmentsupportb)
        ..
        (\tikzinputsegmentlast);
      },
      closepath code={
        \path [#1]
        (\tikzinputsegmentfirst) -- (\tikzinputsegmentlast);
      },
    },
  },
  mid arrow/.style={postaction={decorate,decoration={
        markings,
        mark=at position .5 with {\arrow[#1]{stealth}}
      }}},
}

\newcommand{\cmark}{\color{green}\ding{51}}%
\newcommand{\xmark}{\color{red}\ding{55}}%

\algnewcommand{\IIf}[1]{\State\algorithmicif\ #1\ \algorithmicthen}%
\algnewcommand{\EndIIf}{\unskip\ }%
\algdef{SE}[DOWHILE]{Do}{doWhile}{\algorithmicdo}[1]{\algorithmicwhile\ #1}%
\algnewcommand\algorithmicforeach{\textbf{for each}}%
\algdef{S}[FOR]{ForEach}[1]{\algorithmicforeach\ #1\ \algorithmicdo}%

\newcommand{\code}[1]{{\fontsize{9.5}{11}\selectfont\texttt{#1}}}

\def\ie{{i.e.}}
\def\eg{{\em e.g.}\xspace}

\def\etc{etc.}

\def\namex{\textsc{mind}}
\def\name{\namex\xspace}

\def\mmm{logic and metadata for memory management\xspace}

\def\algo{bounded splitting\xspace}
\def\Algo{Bounded Splitting\xspace}
\def\sizing{\Algo}
\def\fullname{\textbf{M}MU \textbf{I}n-\textbf{N}etwork for \textbf{D}isaggregated architectures}

\begin{document}\sloppy
\title{\name: In-Network Memory Management for Disaggregated Data Centers}

\author{Seung-seob Lee}
\affiliation{%
   \institution{Yale University}
   \country{}}
\author{Yanpeng Yu}
\affiliation{%
   \institution{Peking University}
   \country{}}
\author{Yupeng Tang}
\affiliation{%
  \institution{Yale University}
\country{}}
\author{Anurag Khandelwal}
\affiliation{%
  \institution{Yale University}
\country{}}
\author{Lin Zhong}
\affiliation{%
  \institution{Yale University}
\country{}}
\author{Abhishek Bhattacharjee}
\affiliation{%
  \institution{Yale University}
\country{}}

% acmart requires to define abstract before maketitle
\begin{abstract}
Memory disaggregation promises transparent elasticity, high resource utilization and hardware heterogeneity in data centers by physically separating memory and compute into network-attached resource ``blades''. However, existing designs achieve performance at the cost of resource elasticity, restricting memory sharing to a single compute blade to avoid costly memory coherence traffic over the network. 

In this work, we show that emerging programmable network switches can enable an efficient shared memory abstraction for disaggregated architectures by placing memory management logic \textit{in the network fabric}. We find that centralizing memory management in the network permits bandwidth and latency-efficient realization of in-network cache coherence protocols, while programmable switch ASICs support other memory management logic at line-rate. We realize these insights into \name\footnote{\fullname}, an in-network memory management unit for rack-scale disaggregation. \name enables transparent resource elasticity while matching the performance of prior memory disaggregation proposals for real-world workloads.
\end{abstract}

% keywards and classification
\begin{CCSXML}
<ccs2012>
  <concept>
    <concept_id>10010520.10010521.10010537.10003100</concept_id>
    <concept_desc>Computer systems organization~Cloud computing</concept_desc>
    <concept_significance>500</concept_significance>
  </concept>
  <concept>
    <concept_id>10003033.10003099.10003102</concept_id>
    <concept_desc>Networks~Programmable networks</concept_desc>
    <concept_significance>500</concept_significance>
  </concept>
</ccs2012>
\end{CCSXML}



\maketitle
\pagestyle{plain}

% !TEX root = ../paper.tex
\section{Introduction}
\label{sec:intro}

\begin{comment}
\begin{itemize}[leftmargin=*]
	\itemsep0pt
	\item Disaggregation~\cite{shoal, disagg, disaggnetworks, roce, legoos, infiniswap, decibel, snowset, disaggfault} \& why % (.2 pages)
	\item Memory disaggregation, prior work, challenges % (.3 pages)
	\item Network-driven memory disaggregation: why \& how % (.8 pages, incl. wow example graph)
	\item Contributions % (0.3 pages)
\end{itemize}  
\end{comment}

Data center network bandwidth is approaching that of intra-server resource interconnects~\cite{terabitethernet, remotememory}, and is soon poised to surpass it~\cite{legoosatc}. This has driven significant academic~\cite{memdisagg2, memdisagg3, memdisagg4, memdisagg5, memdisagg6, memdisagg1, legoos, infiniswap, fastswap, disagg, disaggfault} and industry~\cite{industry0, industry1, industry2, industry3, industry4, industry5} interest in memory disaggregation, where compute and memory are physically separated into network-attached \textit{resource blades}, drastically improving resource utilization, hardware heterogeneity, resource elasticity and failure handling compared to traditional data center architectures. 

However, memory disaggregation is challenging due to three requirements. First,  access to remote memory must have low latency and high throughput --- prior work~\cite{legoos, infiniswap, fastswap, disagg} have targeted $10~\mu$s latency and $100$~Gbps bandwidth per compute blade to minimize application performance degradation. Second,  both memory and compute resources available to applications must scale elastically, in keeping with the promise of disaggregation. Finally, wide adoption and immediate deployment requires support for unmodified applications.

Despite years of research towards enabling memory disaggregation, none of the known approaches support all three requirements simultaneously (\S\ref{ssec:challenges}). Most approaches require application modifications due to changes in hardware~\cite{industry1, industry2, nwsupport, memdisagg4}, programming model~\cite{piccolo, grappa}, or memory interface~\cite{farm, ramcloud, herd}. Recent approaches that enable transparent access to disaggregated memory~\cite{legoos, infiniswap, fastswap} limit application compute elasticity --- processes are limited to compute resources on a single compute blade to avoid cache coherence traffic over the network due to performance concerns.

We present \name, the first memory management system for rack-scale memory disaggregation that simultaneously meets all three requirements for disaggregated memory. Our key idea is to place the \mmm \textit{in the network fabric}. \name's design builds on the observation that the network fabric in the disaggregated memory architecture is essentially a CPU-memory interconnect. In \name, centrally-placed in-network processing devices like programmable network switches~\cite{progswitch1, progswitch2, progswitch3} therefore assume the role of the MMU to enable a high-performance shared memory abstraction. Since \name realizes the \mmm in programmable hardware at line rate~\cite{progswitch1}, latency and bandwidth overheads are minimal. 

Realizing in-network memory management, however, requires working with the unique constraints imposed by programmable switch ASICs. First, today's switch ASICs only have a few megabytes of on-chip memory, making it challenging to store traditional page tables for potentially terabytes of disaggregated memory. Second, switch ASICs only permit a few cycles of limited computations per packet to ensure line-rate processing, while cache coherence may require complex state transition logic for each cached block. Finally, these ASICs~\cite{p4paper} have staged packet processing pipelines where compute and memory resources are spread across multiple physically decoupled match-action stages, introducing interesting challenges in partitioning and placing the \mmm across them. 

To meet the three requirements of memory disaggregation, \name effectively navigates the above constraints and explores the capabilities of today's programmable switches to enable in-network memory management for disaggregated architectures. It does so through a principled redesign of traditional memory management:
\begin{itemize}[leftmargin=*, itemsep=0pt]
  \item \name employs a \textit{global virtual address space} shared by all processes, range partitioned across memory blades to minimize the number of address translation entries that need to be stored in the on-chip memory of switch ASIC. At the same time, it employs a physical memory allocation mechanism that load balances allocations across memory blades for high memory throughput (\S\ref{subsec:addr_trans}).
  \item \name features domain-based memory protection inspired by capability-based schemes~\cite{capabilityaddr, cap, opal} that enables fine-grained and flexible protection by decoupling the storage of memory permissions from address translation entries. Interestingly, such a decoupling actually \textit{reduces} the on-chip memory overheads at the switch ASIC (\S\ref{subsec:mem_prot}).
  \item \name adapts directory-based MSI coherence~\cite{msi} to the in-network setting. To mitigate the network overheads of cache coherence, \name exploits network-centric hardware primitives such as multicast in the switch ASIC to efficiently realize its coherence protocol (\S\ref{ssec:caching}).
  \item We find that the limited on-chip memory at the switch ASIC forces the cache directory to track memory regions at coarse granularities, which in turn results in performance degradation due to \textit{false invalidations} of pages in those regions (\S\ref{ssec:caching}). We address this through a novel \sizing algorithm (\S\ref{sec:algorithm}) that dynamically sizes memory regions to bound both the switch storage requirements as well as performance overheads due to false invalidations.
\end{itemize}
\noindent
We realize \name design on a disaggregated cluster emulated using traditional servers connected by a programmable switch. Our results show that \name enables transparent resource elasticity for real-world workloads while matching the performance for prior memory disaggregation proposals (\S\ref{sec:evaluation}). 

We also find that while \name is competitive with compared systems, workloads with high read-write contention experience sub-linear scaling with more threads due to limitations of current hardware. Current x86 architectures preclude realization of relaxed consistency models commonly employed in shared memory systems~\cite{gam}, and the switch TCAM capacity is close to saturated with cache directory entries for such workloads. We discuss approaches that could enable better scaling with future improvements in switch ASIC and compute blade architectures in \S\ref{sec:discussion}.

% !TEX root = ../paper.tex
\section{Background}
\label{sec:background}

This section motivates \name. We discuss key enabling technologies (\S\ref{ssec:assumptions}), followed by challenges in realizing memory disaggregation goals using existing designs (\S\ref{ssec:challenges}).

\paragraphb{Assumptions} We focus on memory disaggregation at the \textit{rack-scale}, where memory and compute blades are connected by a single programmable switch. Similar to prior work~\cite{memdisagg2, memdisagg3, memdisagg4, memdisagg5, memdisagg6, memdisagg1, legoos, disagg}, we restrict our scope to \textit{partial} memory disaggregation: while most of the memory is network-attached, compute blades possess a small amount (few GBs) of local DRAM as cache.

\begin{figure}[t]
  \centering
  \includegraphics[width=0.44\columnwidth]{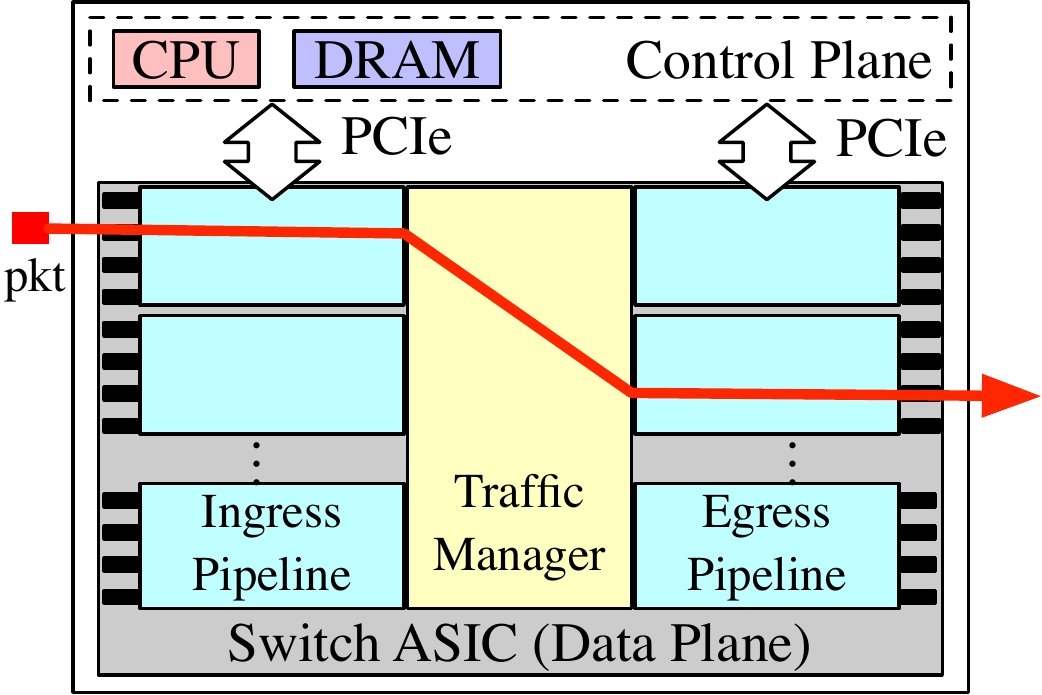}\hfill
  \includegraphics[width=0.54\columnwidth]{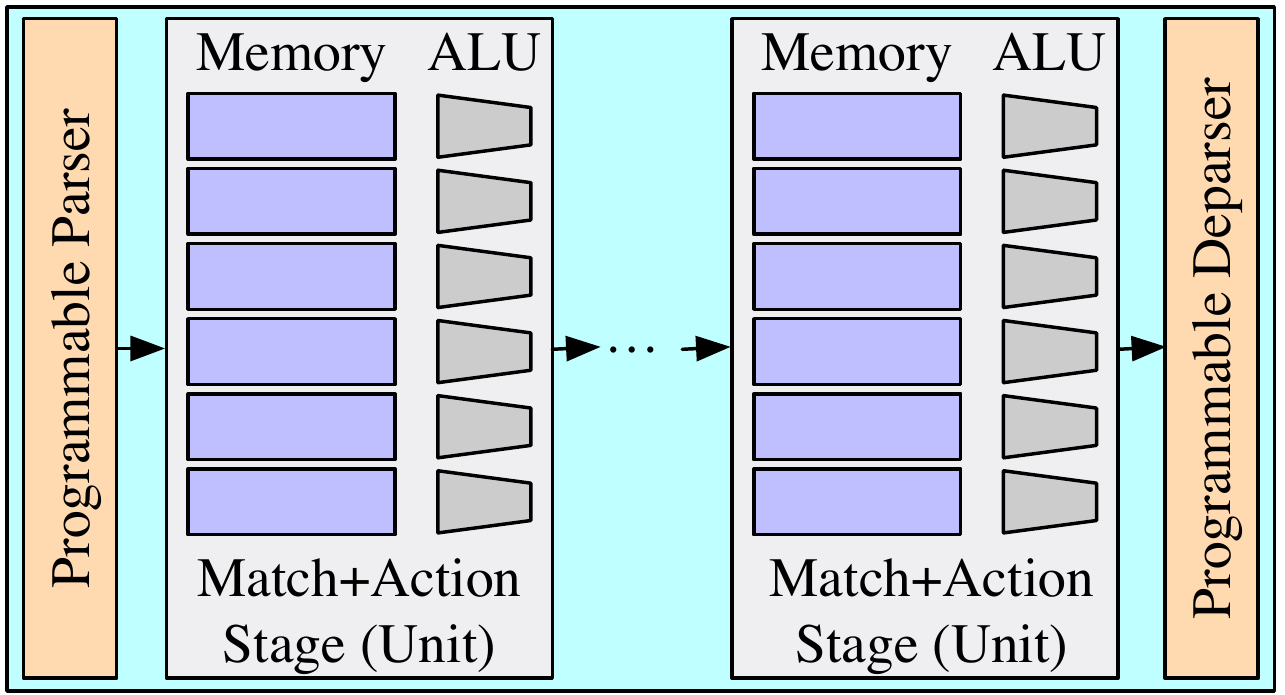}
  \caption{\textbf{Enabling technologies for \name.} (left) Programmable switch architecture and (right) Switch ingress/egress pipeline.}
  \label{fig:trad-directory}
  \label{fig:prog-pipeline}
  \label{fig:background}
\end{figure}

\begin{table}
  \caption{\textbf{In-network technology tradeoffs.} See \S\ref{ssec:assumptions} for details.}\vspace{-0.5em}
  \label{fig:switchtradeoff}
  \scriptsize
    \renewcommand{\arraystretch}{1.2}
    \begin{tabular}{c|c|c|c|c}
      \hline
      & RMT & FPGA & Custom ASIC & CPU \\\hline\hline
      Line-rate & \cmark & \cmark & \cmark & \xmark \\
      Available & \cmark & \cmark & \xmark & \cmark \ \\
      Low Power & \cmark & \xmark & \cmark & \xmark\\
      Low Cost & \cmark & \xmark & \cmark & \xmark\\
      \hline
    \end{tabular}
  \end{table}

\subsection{Enabling Technologies}
\label{ssec:assumptions}

We now briefly describe \name's enabling technologies.

\paragraphb{Programmable switches} In recent years, programmable switches have evolved along two well-coordinated directions: development of P4~\cite{p4, p4paper, dcp4}, a flexible programming language for network switches, and design of switch hardware that can be programmed with it~\cite{rmt, progswitch2, progswitch3, progswitch4}. These switches host an application-specific integrated circuit (ASIC), along with a general purpose CPU with DRAM, as shown in Figure~\ref{fig:background}~(left). The switch ASIC comprises ingress pipelines, a traffic manager and egress pipelines, which process packets in that order. Programmability via P4 is facilitated through a programmable parser and match-action units in the ingress/egress pipelines, as shown in Figure~\ref{fig:prog-pipeline}~(right). Specifically, the program defines how the parser parses packet headers to extract a set of fields, and multiple stages of match-action units (each with limited TCAM/SRAM and ALUs) process them. The general purpose CPU is connected to the switch ASIC via a PCIe interface, and serves two functions: (i) performing packet processing that cannot be performed in the ASIC due to resource constraints, and, (ii) hosting controller functions that compute network-wide policies and push them to the switch ASIC.

While the above focuses on switch ASICs with Reconfigurable Match Action Tables (RMTs)~\cite{rmt}, it is possible to realize \name using FPGAs, custom ASICs, or even general purpose CPUs. While each of them exposes different tradeoffs (Table~\ref{fig:switchtradeoff}), we adopt RMT switches due to their performance, availability, power and cost efficiency.

\paragraphb{DSM Designs} Traditionally, shared memory has been explored in the context of NUMA~\cite{sgiorigin, amdopteron1, amdopteron2, intelqpi1, intelqpi2} and distributed shared memory (DSM) architectures~\cite{munin, midway, dpram, gam, dash}. In such designs, the virtual address space is partitioned across the various nodes, \ie, each partition has a \textit{home} node that manages its metadata, \eg, the page table. Each node also additionally has a cache to facilitate performance for frequently accessed memory blocks. We distinguish memory blocks from pages since caching granularities, \ie, block, can be different from memory access granularities, \ie, page. 

With the copies of blocks potentially residing across multiple node caches, coherence protocols~\cite{msi, mesi, mesif, moesi, mosi} are required to ensure each node operates on the latest version of a block. In popular directory-based invalidation protocols like MSI~\cite{msi} (used in \name), each memory block can be in one of three states: \textbf{M}odified (\textbf{M}), where a single node has exclusive read and write access to (or, ``owns'') the block, \textbf{S}hared (\textbf{S}), where one or more caches have shared read-only access to the block, and \textbf{I}nvalid (\textbf{I}), where the block is not present in any cache. A directory tracks the state of each block, along with the list of nodes (``sharer list'') that currently hold the block in their cache. The directory is typically partitioned across the various nodes, with each home node tracking directory entries for its own address space partition. Memory access for a block that is not local involves contacting the home node for the block; it triggers a state transition and potential invalidation of the block across other nodes, followed by retrieving the block from the node that owns the block. While it is possible to realize more sophisticated coherence protocols, we restrict our focus to MSI in this work due to its simplicity --- we defer a discussion of other protocols to \S\ref{sec:discussion}.

\subsection{Disaggregated Memory Designs and Challenges}
\label{ssec:challenges}

\noindent
As outlined in \S\ref{sec:intro}, extending the benefits of resource disaggregation to memory and making them widely applicable to cloud services demands (i) low-latency and high-throughput access to memory, (ii) a transparent memory abstraction that supports elastic scaling of memory \textit{and} compute resources without requiring modifications to existing applications. Unfortunately, prior designs for memory disaggregation expose a hard tradeoff between the two goals. Specifically, transparent elastic scaling of an application's compute resources necessitates a shared memory abstraction over the disaggregated memory pool, which imposes non-trivial performance overheads due to the cache-coherence required for both application data \textit{and} memory management metadata. We now discuss why this tradeoff is fundamental to existing designs. We focus on page-based memory disaggregation designs here, and defer the discussion of other related work to \S\ref{sec:related}.

\paragraphb{Transparent designs} While transparent DSMs have been studied for several decades, their adaptation to disaggregated memory has not been explored. We consider two possible adaptations for the approach outlined in \S\ref{ssec:assumptions} to understand their performance overheads, and shed light on why they have remained unexplored thus far. The first is a \textit{compute-centric} approach, where each compute blade owns a partition of the address space and manages the corresponding metadata, but the memory itself is disaggregated. A compute blade must now wait for several sequential remote requests to be completed for every un-cached memory read or write, \eg, to the remote home compute blade to trigger state transition for the block and invalidate relevant blades, and to fetch the memory block from the blade that currently owns the block. An alternate \textit{memory-centric} design that places metadata at corresponding home memory blades still suffers multiple sequential remote requests for a memory access as before, with the only difference being the home node accesses are now directed to memory blades. While these overheads can be reduced by caching the metadata at compute blades, it necessitates coherence for the metadata as well, incurring additional design complexity and performance overheads.

\paragraphb{Non-transparent designs} Due to the anticipated overheads of adapting DSM to memory disaggregation, existing proposals limit processes to a single compute blade~\cite{industry0, industry1, memdisagg1, infiniswap, legoos, disagg}, \ie, while compute blades cache data locally, different compute blades do not share memory to avoid sending coherence messages over the network. As such, these proposals achieve memory performance only by limiting transparent compute elasticity for an application to the resources available on a single compute blade, requiring application modifications if they wish to scale beyond a compute blade. 

% !TEX root = ../paper.tex

\begin{table}
    \caption{\small \textbf{Parallels between memory \& networking primitives.}}\vspace{-0.5em}
    \label{table:isomorph}
    \centering
    \scriptsize
    \renewcommand{\arraystretch}{1.2}
    \begin{tabular}{p{2cm} p{0.7cm}p{2cm}}
      \hline
      \textbf{Virtual Memory} &$\Longleftrightarrow$ &\textbf{Networking} \\\hline\hline
      Memory allocation&&IP assignment\\
      Address translation &&IP forwarding\\
      Memory protection  &&Access control\\
      Cache invalidations &&Multicast\\
      \hline
    \end{tabular}
\end{table}

\section{\name Overview}
\label{sec:overview}

\begin{figure*}[!t]
\centering
\includegraphics[width=0.55\textwidth]{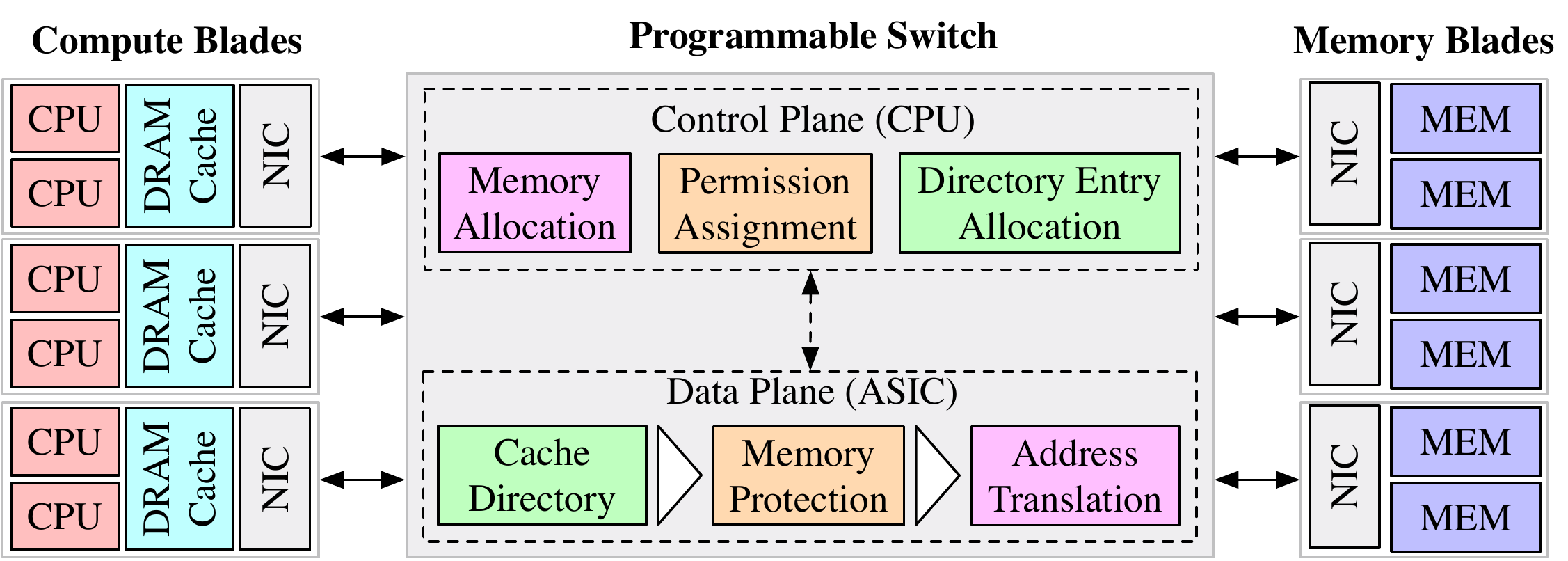}\hspace{3em}
\includegraphics[width=0.35\textwidth]{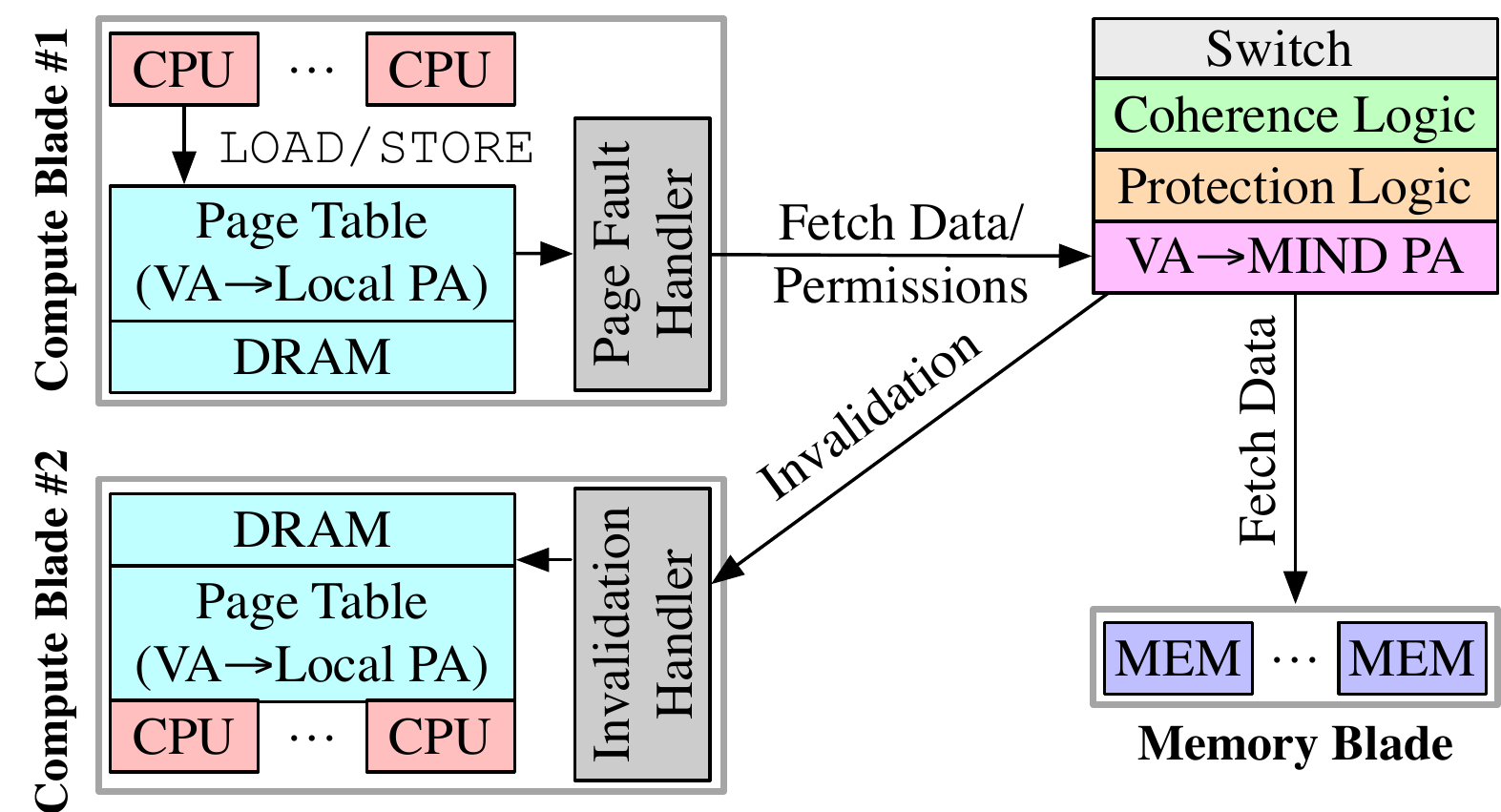}%\vspace{-1em}
\caption{\textbf{(left) High-level \name architecture, and, (right) data flow for memory accesses in \name.} See \S\ref{ssec:design} for details.}
\label{fig:system_diagram}
\end{figure*}

To break the tradeoff highlighted above, we place memory management \textit{in the network fabric} for three reasons.
First, the network fabric enjoys a central location in the disaggregated architecture. Therefore, placing memory management in the data access path between compute and memory resources obviates the need for metadata coherence. 
Second, modern network switches~\cite{progswitch1, progswitch2, progswitch3} permit the implementation of such logic in integrated programmable ASICs. We show that these ASICs are capable of executing it at line rate even for multi-terabit traffic. In fact, many memory management functionalities have similar counterparts in networking (\autoref{table:isomorph}), allowing us to leverage decades of innovation in network hardware and protocol design for disaggregated memory management.
Finally, placing the cache coherence logic and directory in the network switch permits the design of specialized in-network coherence protocols with reduced network latency and bandwidth overheads, as we show in \S\ref{sec:design}. 

Effective in-network memory management requires: (\text{i}) \emph{efficient storage}, by  minimizing in-network metadata given the limited memory on the switch data plane;  (\textit{ii}) \emph{high memory throughput}, by load-balancing memory traffic across memory blades; (\textit{iii}) \emph{low access latency to shared memory}, via efficient cache coherence design that hides the network latency.

Next we elicit  three design principles followed by \name to realize the above goals and provide an overview of its design.

\subsection{Design Principles}
\label{ssec:principles}

\name follows three principles to meet the goals for memory disaggregation outlined in \S\ref{sec:intro}:

\paragraphb{P1} \textit{Decouple memory management} functionalities to ensure each can be optimized for their specific goals.

\paragraphb{P2} Leverage \textit{global view} of the disaggregated memory subsystem at a centralized control plane to compute optimal policies for each memory management functionality.

\paragraphb{P3} Exploit \textit{network-centric hardware primitives} at the programmable switch ASIC to efficiently realize policies computed using \textbf{P2}.

\vspace{0.075in}\noindent

\name follows principle \textbf{P1} to decouple memory allocation from addressing (\S\ref{subsec:addr_trans}), address translation from memory protection (\S\ref{subsec:mem_prot}),  cache accesses and eviction from coherence protocol execution (\S\ref{subsec:cache_dir}),  and employs principles \textbf{P2} and \textbf{P3} to efficiently realize their goals. Note that traditional server-based OSes are unable to leverage these principles due to  their reliance on \textit{fixed-function} hardware modules such as the MMU and memory controller --- most common implementations of such modules couple many memory management functionalities (\eg, address translation and memory protection in page-table walkers) for a host of complexity, performance, and power reasons~\cite{vmbook, seesaw, rmmlite}.

\subsection{Design Overview}
\label{ssec:design}

\name exposes a \textit{transparent virtual memory} abstraction to applications, similar to server-based OSes. Unlike prior disaggregated memory designs, \name places all \mmm in the network, instead of compute or memory blades~\cite{infiniswap, fastswap}, or a separate global controller~\cite{legoos}. 

Figure~\ref{fig:system_diagram}~(left) provides an overview of \name design, while Figure~\ref{fig:system_diagram}~(right) depicts the data flow for memory accesses in \name. The \textit{compute blades} run user processes and threads, and possess a small amount of local DRAM that is used as a cache. All memory allocations (\eg, via \texttt{mmap} or \texttt{sbrk}) and deallocations (\eg, via \texttt{munmap}) from the user processes are intercepted at the compute blade, and forwarded to the \textit{switch control plane}. The control plane possesses a global view of the system, which it leverages to perform memory allocations, permission assignments, \etc, using principle \textbf{P2}, and respond to the user process. All memory \code{LOAD}/\code{STORE} operations from the user processes are handled by the compute blade cache (\S\ref{ssec:caching}). The cache is virtually addressed\footnote{Note that while it is hidden from applications, compute blades maintain a local page-based virtual memory abstraction to translate \name virtual addresses to physical addresses for cached pages in local DRAM (Figure~\ref{fig:system_diagram}~(right)).}, and stores permissions for cached pages to enforce memory protection. If a page is not locally cached, the compute blade triggers a page-fault and fetches the page from \textit{memory blades} using RDMA requests, evicting other cached pages if necessary. Similarly, if the memory access requires an update to a cached block's coherence state (\eg, \code{STORE} on a Shared or \textbf{S} block), a page-fault is triggered to initiate cache coherence logic at the switch. Note that the page-fault based design requires \name to perform page-level remote accesses, although future CPU architectures may enable more flexible access granularities (\S\ref{sec:discussion}).

Since the compute blade does not store memory management metadata, the RDMA requests are for virtual addresses and do not contain endpoint information (\eg, IP address) for the memory blade that holds the page. Consequently, the \textit{switch data plane} intercepts these requests. It then performs necessary cache coherence logic, including lookups/updates to the cache directory and cache invalidations on other compute blades (\S\ref{ssec:caching}, \S\ref{sec:algorithm}). In parallel, the data plane also ensures the requesting process has permissions to access the requested page (\S\ref{subsec:mem_prot}). If no compute blade cache holds the page, the data plane translates the virtual addresses to physical addresses (\S\ref{subsec:addr_trans}), forwarding the request to the appropriate memory blade. These memory management functionalities are decoupled as separate modules following \textbf{P1}, and efficiently realized in the switch ASIC following \textbf{P3}.

In \name's design, the memory blades simply store the actual memory pages, and serve RDMA requests for physical pages. Unlike prior works that employ RPC handlers and polling threads~\cite{legoos}, \name leverages one-sided RDMA operations~\cite{farm} to obviate the need for any CPU cycles on the disaggregated memory blades. This is a step towards true hardware resource disaggregation, where memory blades need no longer be equipped with any general-purpose CPUs.

% !TEX root = ../paper.tex
\section{In-Network Memory Management}
\label{sec:design}

Placing memory management logic and metadata in the network provides the opportunity for simultaneously achieving memory performance and resource elasticity. We now describe how \name optimizes for the individual goals of memory allocation and addressing (\S\ref{subsec:addr_trans}), memory protection (\S\ref{subsec:mem_prot}), and cache coherence (\S\ref{ssec:caching}), while operating under the constraints of programmable switches. Finally, we detail how \name handles failures (\S\ref{subsec:acking}). %We now describe how \name achieves the same.

% !TEX root = ../paper.tex
\subsection{Memory Allocation \& Addressing}
\label{subsec:addr_trans}

Traditional virtual memory uses fixed sized pages as the basic units for both translation and protection; as a result, it cannot achieve the goal of storage efficiency without increasing memory fragmentation: small pages reduce memory fragmentation but require more translation entries, and vice versa.  Following \textbf{P1}, \name overcomes this by \textit{decoupling} address translation and protection.  That is, \name's translation is blade-based while protection is \code{vma} based (\S\ref{subsec:mem_prot}).

\paragraphb{Storage-efficient address translation} \name eschews page-based protection but uses a \textit{single global virtual address-space} across all processes, allowing translation entries to be shared across them. Our approach builds on decades of research on virtual memory designs that also exploit a single address space~\cite{cheri, cap, gam, grappa, opal}, but adds techniques to minimize storage overheads for in-network address translation.
In particular, \name \textit{range partitions} the virtual address space across different memory blades, such that the entire virtual-address space maps to a contiguous range of physical address space. This allows us to use a single translation entry for each memory blade: any virtual address that falls within its range can be directly routed to that memory blade, minimizing the storage required on switch data plane. In \name, this mapping only changes when new memory blades join, old ones retire or if memory is moved between blades.

\paragraphb{Balanced memory allocation \& reduced fragmentation} \name's control plane, leveraging its global view of allocations (\textbf{P2}), tracks the total amount of memory allocated on each memory blade and places a new allocation on the blade with the least allocation, to achieve near-optimal load-balancing. We validate this empirically in \S\ref{sec:evaluation}.

Moreover, since there is a one-to-one mapping between virtual and physical addresses within a particular memory blade, \name minimizes external fragmentation at each memory blade by using traditional virtual memory allocation schemes that have evolved to facilitate the same, \eg, first-fit allocator in our implementation~\cite{firstfit}.
The result of memory allocation is a virtual memory area (\texttt{vma}), identified by the base virtual address and length of the area, \eg, 
  {{\small <\texttt{0x00007f84b862d33b}, \texttt{0x400}>}}
for a $1$KB area. 
As will be elaborated in \S\ref{subsec:mem_prot}, \code{vma} is the basic unit of protection in \name.
This allows multiple processes to have non-overlapping \code{vma}s on the same blade, minimizing memory fragmentation.
 
\paragraphb{Isolation} We note that \namex's global virtual address-space does not compromise on \textit{isolation} between processes. First, since the switch intercepts allocation requests across all compute blades, and possesses a global view of valid allocations at any time, it can easily ensure allocations are non-overlapping across different processes. Second, we show in \S\ref{subsec:mem_prot} that \name's \code{vma}-based protection allows flexible access control between processes in a single global virtual address-space.

\paragraphb{Transparency via outlier entries} \namex's one-to-one mapping between virtual and physical addresses does not preclude supporting unmodified applications with static virtual addresses embedded within their binaries, or OS optimizations such as page migration~\cite{pagemigrations}, \ie, moving pages from one memory blade to another. \name maintains separate \textit{range-based} address translations~\cite{rangetranslations} for physical memory regions that correspond to static virtual addresses or migrated memory. These \textit{outlier} entries are stored succinctly in the switch TCAM, where the TCAM's longest-prefix matching (LPM) property ensures that only the most specific entry (\ie, one with the longest prefix) is considered when translating a virtual address, ensuring correctness.

% !TEX root = ../paper.tex

\subsection{Memory Protection}
\label{subsec:mem_prot}

As \name decouples translation and protection, it uses a separate table to store memory protection entries in the data plane. 
Consequently, an application can assign access permissions to a \code{vma} of any size.
The size of this protection table is proportional to the number of \code{vma}s. We find this number is reasonably small in our experiments and the protection table can easily fit in the switch ASIC even for a wide range of memory-intensive applications (\S\ref{sec:evaluation}). This is because the first-fit allocator and Linux's \code{glibc} allocation requests~\cite{glibc-alloc} do a good job of ensuring \code{vma}s are large and contiguous. 

\paragraphb{Fine-grained, flexible memory protection} Similar to prior work on capability-based systems~\cite{cheri, capabilityaddr}, \name supports two key abstractions: \textit{protection domains} and \textit{permission classes}. Protection domains identify the entity that may (or may not) have permissions to access a particular memory region of arbitrary size, while the permission class identifies what the entity can do to the memory region. 
\name's control plane exposes a set of APIs for memory allocation and permission changes that allows an application to specify a protection domain identifier (\code{PDID}) for an arbitrary virtual memory area (\code{vma}) and assign a permission class (\code{PC}) to the pair \code{<PDID}, \code{vma>}. 
The mapping \code{<PDID}, \code{vma>} $\rightarrow$ \code{PC} is stored as an entry in the protection table in the data plane. For existing applications, \name simply takes the process identifier (\code{PID}) as the \code{PDID},  and uses Linux memory permissions (\eg, read-only, read-write, \etc) as permission classes. Note that \name \textit{can} support richer memory protection semantics than traditional OSes, \eg, user programs that serve multiple client sessions, such as ssh servers or database services, can assign a separate protection domain per session to prevent one session from accessing data from other sessions~\cite{cheri}. 

Following principle \textbf{P3}, we leverage TCAM-based parallel range matches in the programmable switch ASIC --- typically used for IP subnet matches --- to efficiently support fine-grained matching for \code{<PDID}, \code{vma>} entries embedded in memory access requests and obtain corresponding the permission class (\code{PC}). If there is a mismatch between \code{PC} and the memory access type, or the \code{<PDID}, \code{vma>} entry does not exist, the request is rejected.

\paragraphb{Optimizing for TCAM storage} One limitation of TCAM is that each of its entries can only match power-of-two ranges. \name overcomes this by splitting an arbitrary-sized virtual address range into multiple power-of-two-sized entries. Note that the number of entries required for a range of size $s$ is upper-bounded by $\lceil\log_2(s)\rceil$. In order to meet our goal of storage efficiency in the switch data plane, the control plane (1) only performs virtual address allocations that are aligned with the power-of-two sizes to ensure each region can be represented using a single TCAM entry, and (2) coalesces adjacent entries with if they belong to the same protection domain and have the same permission class. Interestingly, memory allocations requested by underlying libraries (\eg, \texttt{glibc}) are mostly in power-of-two sizes anyway, enabling storage-efficiency for TCAM entries.
% !TEX root = ../paper.tex

\subsection{Caching \& Cache Coherence}
\label{ssec:caching}

In \name design, while the caches\footnote{Note that we use the term `cache' to refer to the DRAM at the compute blade under the partial disaggregation model, and not hardware (L1/L2/L3) caches.} reside on compute blades, the coherence directory and logic reside in the switch. This already permits access to the cache directory in half a round-trip, significantly reducing the latency overheads for the coherence protocol execution. For MSI protocol, even the most expensive and relatively uncommon transitions (\ie, \textbf{M}$\rightarrow$\textbf{S/M}) incur two round-trips, while common transitions incur only a single round-trip, as we show in \S\ref{ssec:bottlenecks}. While performance is a primary objective in \name's cache coherence, the coherence protocol must also be realizable under the compute and memory constraints of switch ASICs. We now outline challenges in adapting traditional cache management to our in-network setting, along with how \name resolves them.

\subsubsection{Storage vs. performance tradeoff}\label{ssec:storagevsperf}\hfill\\
Traditional caching and cache coherence mechanisms applied to \name expose a tradeoff between cache performance and the storage efficiency at the switch data plane. Specifically, reducing the number of directory entries requires larger cache granularities (\ie, larger memory blocks), which results in worse performance. For instance, when large (\eg, $2$ MB) memory blocks are used, updating a small (\eg, $4$ KB) region within the block will invalidate the entire block. We refer to these invalidations as \emph{false invalidations} --- dirty pages invalidated along with the requested page because there are in the same memory block tracked by a directory entry. This leads to wastage in both memory bandwidth and cache capacity, \ie, fewer frequently accessed data items in the cache. We empirically highlight this tradeoff in \S\ref{ssec:sensitivity}.

\name addresses this challenge using two approaches: it decouples the cache and directory granularities (following principle \textbf{P1}), and appropriately sizes the memory region tracked by each cache directory entry leveraging the global view of memory traffic at the control plane (following principle \textbf{P2}), as we describe next.

\paragraphb{Decoupling cache access \& directory entry granularities} Our first approach employs principle \textbf{P1} --- decoupling the granularity of cache (and memory) accesses from the granularity at which cache coherence is performed. This allows memory accesses (\eg, evictions or remote memory reads) to be performed at finer granularities, while directory entries are tracked at coarser granularities. Specifically, accesses to the local DRAM cache at compute blades, and even the movement of data between those DRAM caches and memory blades, occur at the fine page granularity ($4$ KB in \name, similar to prior work~\cite{legoos, infiniswap, fastswap}). However, the coherence protocol tracks directory entries (stored at the switch data plane) at larger, variable-sized \textit{region} granularities --- when a $4$ KB page is cached at a compute blade, \name creates a directory entry for the region that contains the page. An invalidation of the region triggers an invalidation of all dirty pages in the region, as tracked by the individual compute blades that cache them.

\paragraphb{Storage \& performance-efficient sizing of \textit{regions}} Even with the decoupling described above, the region \textit{sizes} still expose a tension between coherence performance (\eg, larger false invalidation counts due to larger region sizes) and directory storage efficiency (\eg, more directory entries due to smaller region sizes). To appropriately size regions, \name leverages global view of memory traffic at the switch control plane (\textbf{P2}). Briefly, \name starts each directory entry with a very large memory region; when the overhead due to false invalidation is high, it splits the region and creates a new directory entry. It does so repeatedly, until either the overhead is below a predetermined threshold, or the region size reaches $4$~KB, \ie, the page size. In doing so, \name dynamically customizes the region sizes to resolve the tension between performance and directory storage efficiency using a novel \sizing algorithm --- we defer its details to \S\ref{sec:algorithm}.

\subsubsection{In-Network Coherence Protocol}\label{subsec:cache_dir}\hfill\\
\noindent
Due to the limited computational capability at the switch ASIC, \name employs the simple directory-based MSI coherence protocol~\cite{msi}. While we defer the implementation details of the coherence protocol to \S\ref{ssec:switchimpl}, we highlight here how \name employs network-centric hardware primitives to efficiently realize the coherence protocol in the switch, leveraging principle \textbf{P3}. Specifically, several state transitions in the MSI protocol require generating invalidation requests to compute blades that have shared access to a region, to ensure correctness. To facilitate this in a network-efficient manner, we leverage \textit{multicast} functionality supported natively in most switches --- we create a multicast group for all compute blades in the rack, and send an invalidation request containing the list of sharers to the group. However, broadcasting invalidations to blades not in the sharer list would consume unnecessary network bandwidth. As such, we embed the sharer list within the invalidation request, and drop requests in the egress path of the switch data plane if the output port does not lead to a blade in the sharer list.

\subsection{Handling Failures}\label{subsec:acking}

We now discuss how \name handles failures at different components in our disaggregated architecture.

\paragraphb{Compute, memory blade and switch failures} \name does not innovate on fault-tolerance for compute and memory blade failures: mechanisms developed in prior work~\cite{legoos, infiniswap, disaggfault} for fault tolerance can be readily adapted to our design. To handle switch failures, we consistently replicate the control plane at a backup switch --- on a failure, the data plane state is reconstructed at the backup switch using the control plane state. Since the control plane is only updated infrequently due to metadata operations (\eg, system calls), the overhead added due to such replication in minimal.

\paragraphb{Communication failures} \name uses ACKs and timeouts to detect packet losses. When a memory access triggers invalidations, the requesting compute blade waits for ACKs indicating successful invalidation from all sharers, and resends the request if timeout occurs. If the compute blade does not receive an ACK even after a predefined number of retransmissions, it sends a \code{reset} message for the corresponding virtual address to the switch control plane. This, in turn, forces all compute blades to flush their data for that address and removes the corresponding cache directory entry in the data plane. This \code{reset} mechanism prevents deadlocks when compute blades fail in the middle of a cache coherence state transition.

% !TEX root = ../paper.tex

\section{\Algo: Algorithm \& Analysis}
\label{sec:algorithm}

We next provide details and a formal analysis of the \algo algorithm used by \name to dynamically determine the memory region size tracked by each cache directory entry.  This algorithm is a key component of \name's cache coherence design outlined in \S\ref{ssec:caching}.

\subsection{Algorithm}
\label{ssec:detail}

The \algo algorithm starts by partitioning the entire virtual address space into $N$ contiguous regions of size $M$ pages each. It then works in disjoint \emph{epochs} of equal length. In each epoch, it tracks the total number of times any page is falsely invalidated within the epoch  --- we refer to this as the \textit{false invalidation count} --- for every region.

\Algo uses false invalidation count as a measure of the performance overhead --- we characterize the impact of false invalidations on performance in \S\ref{sec:evaluation}. It therefore seeks to keep the count below a threshold, denoted by $t$. 
If any region has a false invalidation count $> t$ in an epoch, it splits that region into two equal halves and creates a new directory entry accordingly. It bounds the smallest size of any memory region to $4$KB (the page size), ensuring that any $M$ sized block is split at most $\log_2{M}$ times over as many epochs. Note that for 4KB regions, the number of falsely invalidated pages is trivially zero; however, maintaining all regions at that size would require storing $N\cdot M$ directory entries at the switch, which is impractically large.

\paragraphb{Stability assumptions} The \algo is based on two implicit assumptions related to the epoch length:
\begin{itemize}[itemsep=0pt, leftmargin=*]
  \item the access pattern across the various memory regions remains stable for at least $O(\log_2{M})$ epochs, and
  \item the set of allocated pages remains unchanged across $O(\log_2{M})$ epochs.
\end{itemize}

\noindent
We show in \S\ref{sec:evaluation} that appropriate epoch sizing allows us to ensure both assumptions for our evaluated workloads.

\subsection{Performance Bounds}
\label{ssec:bound}
The key challenge in \algo is bounding the number of directory entries that must be stored, one per memory region. The total \textit{worst-case} number of regions depends on two factors: (i) worst-case number of \textit{sub}-regions generated by each $M$-sized region, and (ii) the value of $t$. We first analyze the worst-case bound on the number of sub-regions per $M$-sized region, and then bound the worst-case for total number of regions overall (and therefore, the total number of directory entries) by appropriately setting the value of $t$.

\paragraphb{Bounding the number of sub-regions per $M$-sized region} We establish the worst-case bound in the following theorem:

\begin{myTheorem}\label{theorem:region}
  The number of sub-regions for an $M$-sized region with false invalidation count $f$ is upper-bounded by $S=(\lceil\frac{f}{t}\rceil-1)\cdot(1+\log_2{M})$.
\end{myTheorem}
\begin{proof}
  In the \algo algorithm, we dynamically manage sizes of each memory region, splitting it into two smaller regions until the false invalidation count for the region is less than $t$. As noted above, since we limit the smallest region size to 4KB, an $M$-sized region may be split at most $\log_2{M}$ times, across as many epochs. Figure~\ref{fig:cache_tree} depicts the splitting process as a binary tree across the various epochs, where the level of the tree $l$ denotes the epoch index ($0 \leq l < \log_2{M}$). For the sake of exposition, we set $M=2$MB. We leverage two observations to prove the above bound:
  \begin{itemize}[itemsep=0pt, leftmargin=*]
    \item{\textbf{O1:}} Splitting a region can only \textit{decrease} the false invalidation count across the two splits, \ie, if a region with $f$ false invalidation count is split into two regions with $f'$ and $f''$ false invalidation count, then $f' + f'' \leq f$.
    \item{\textbf{O2:}} For a 4KB region, the false invalidation count is zero.
  \end{itemize}  
  
\begin{figure}[t]
  \centering
  \includegraphics[width=0.97\columnwidth]{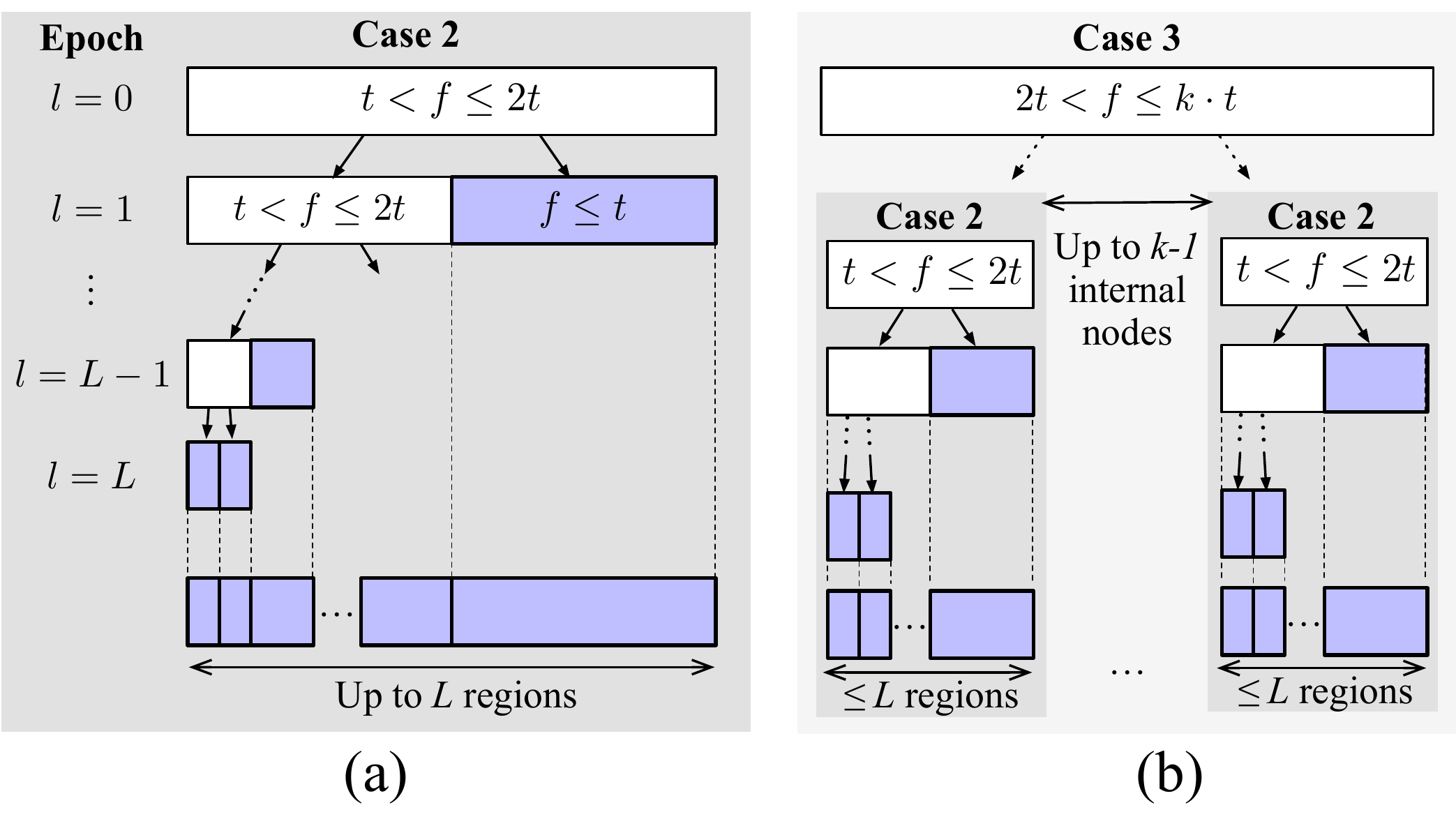}
  \caption{\textbf{Splitting process for cache blocks depicted as a binary tree.} Note that $L=\log_2{M}$; see 
  \S\ref{ssec:bound} for details.}\label{fig:cache_tree}
\end{figure}
  
  \noindent
  The maximum number of regions $S$ generated by splitting an $M$-sized region can be categorized into three cases:
  
  \paragraphb{Case 1: $f \leq t$} Since false invalidation count is already below the threshold, the region does not need to be split, \ie, $S=1$.
  
  \paragraphb{Case 2: $t < f \leq 2t$} To bring the false invalidation count below $t$, the region will be split into two. Due to observation \textbf{O1}, there are two possibilities: (i) both resulting regions have false invalidation count $<t$, or (ii) one region $r_1$ still has false invalidation count $>t$ while the other region $r_2$ has false invalidation count $<t$. For (i), the resulting regions do not need to be split any further, while for (ii), $r_1$ must be split further in the next epoch. In the worst case, the splits will continue for at most $\log_2{M}$ epochs --- when the region size reaches 4KB, no further splits will be required (due to observation \textbf{O2}). Thus, $S = 1 + \log_2{M}$, as shown in Figure~\ref{fig:cache_tree}~(a).
  
  \paragraphb{Case 3: $2t < f \leq k\cdot t$, where $k = \lceil\frac{f}{t}\rceil$} The worst-case scenario that maximizes the number of generated regions must maximize the number of ``internal nodes'' that can generate such regions in the binary tree depicting the splitting process. In particular, the scenario should create as many internal node regions with false invalidation count between $t$ and $2t$ as possible, and then employ \textbf{Case 2} to maximize the number of final regions generated by each ``internal node'' region. Note that for $2t < f < k\cdot t$, the region may be split into at most $k - 1$ regions where each region has false invalidation count between $t$ and $2t$ (regardless of how many epochs it takes). With the worst-case number of regions generated by each such internal node region given in \textbf{Case 2}, the upper bound on the number of generated regions is given by (Figure~\ref{fig:cache_tree}~(b)):
  $$S = (k - 1) \cdot (1 + \log_2{M}) = (\lceil\frac{f_i}{t}\rceil - 1) \cdot (1+\log_2{M})$$
  \noindent
  As such, across the three cases, the total number of regions is at most $(\lceil\frac{f}{t}\rceil - 1) \cdot (1+\log_2{M})$.
\end{proof}

\paragraphb{Bounding the total number of regions} We now consider the worst-case number of regions $S_{max}$ contributed by \textit{all} $M$-sized regions. Let $f_i$ be the number of false invalidation count for an $M$-sized region $i$ ($1 \leq i \leq N$), and $S_i$ be the worst-case number of regions generated by it, then:
$$ S_{max} = \sum\limits_{i=1}^{N} S_i = \sum\limits_{i=1}^{N} (\lceil\frac{f_i}{t}\rceil - 1) \cdot (1+\log_2{M}) $$
\noindent
To bound $S_{max}$, we must set $t$ appropriately. In order to ensure fairness across all $M$-sized regions, we set the threshold $t$ as a fraction of the average false invalidation across them, \ie, 
\begin{align}\label{eq:t}
  t = \frac{1}{c\cdot N} \cdot \sum\limits_{i=1}^{N} f_i
\end{align}
\noindent
where $c$ is a constant parameter.

This allows us to bound $S_{max}$ as follows:
\begin{align*}
  S_{max} &= \sum\limits_{i=1}^{N} (\lceil\frac{f_i}{t}\rceil - 1) \cdot (1+\log_2{M}) \leq \sum\limits_{i=1}^{N} \frac{f_i}{t} \cdot (1+\log_2{M})                &\\
          &= c \cdot N \cdot (1+\log_2{M}) \text{$\qquad$(From Eq.~\ref{eq:t})}
\end{align*}
\noindent
If use up all the available switch data plane capacity to store $S_{max}$ entries, we can set $c$ as $\frac{S_{max}}{N\cdot(1+\log_2{M})}$, which will always ensure the total number of regions is $\leq S_{max}$.

\paragraphb{Split vs. merge-based approach} The approach we have described so far starts with $M$-sized regions, and splits into regions until the false invalidation count for each region reduces below the threshold $t$. An alternate but equivalent strategy would begin with 4KB regions and merge them into larger regions as long as the false invalidation count per region remains below $t$. In fact, it is possible to begin with any intermediate region size, and split or merge as necessary. In \name, we use a default of $16$KB since it provides a favorable tradeoff between storage and performance overheads for our evaluated workloads --- we defer a detailed analysis to \S\ref{sec:evaluation}.

\paragraphb{From theory to practice} At $c=1$, the dynamic resizing approach outlined above reduces the amount of storage required for directory entries from $M\cdot N$ to a worst-case of $(1+\log_2{M})\cdot N$ --- an \textit{exponential} decrease. At the same time, it ensures that the number of false invalidation count remains under $\frac{\sum f_i}{N}$. However, we note that our theoretical analysis only reveals the \textit{worst-case} --- in practice, we find both storage and performance overheads are much lower, as we show in \S\ref{sec:evaluation}. As such, we can set the value $c > 1$ to increase switch data plane storage utilization without hitting its capacity in practice. In fact, we dynamically adjust the value of $c$ such that the utilization of the switch data plane storage in any epoch remains below $95\%$.

% !TEX root = ../paper.tex
\section{Implementation Details}
\label{sec:impl}

We now describe \name implementation. \name exposes Linux memory and process management system call APIs, and splits its kernel components across compute blade and the programmable switch. We now describe these kernel components, along with the RDMA logic required at the memory blade.

\subsection{Compute Blade}
\label{ssec:cpumemimpl}

\name assumes a partial disaggregation model, where the compute blades possess a small amount of local DRAM as cache (\S\ref{ssec:assumptions}). The compute blades in our prototype use traditional servers with no hardware modifications. We implemented compute blade kernel components as a modified Linux kernel 4.15. \name provides transparent access to the disaggregated memory, by modifying how \code{vma}s and processes are managed and how page faults are handled at the compute blade, as we detail next.

\paragraphb{Managing \code{vma}s} To handle the creation and removal of \code{vma}s due to process heap allocation/deallocation requests, such as \code{brk}, \code{mmap}, and \code{munmap}, the kernel module intercepts such requests from the process and forwards them to the control plane at the switch over a reliable TCP connection. The switch subsequently creates new \texttt{vma} entries, and responds with the same return value (\eg, virtual address of the allocated \code{vma}) as the local version of the system calls --- ensuring transparency for user applications. The switch returns Linux-compatible error codes (\eg, \code{ENOMEM}) if there are any errors.

\paragraphb{Managing processes} The kernel module also intercepts and forwards process creation and termination requests, such as \code{exec} and \code{exit}, to the switch control plane, which maintains the internal representation of processes (\ie, Linux's \code{task\_struct}) and a mapping between the compute blades and processes they host. \name assigns threads running on different compute blades the same PID if they belong to the same process, permitting them to transparently share the same address-space via memory protection and address translation rules installed at the switch. Finally, we do not focus on scheduling in this work and simply place threads and processes across compute blades in a round-robin manner.

\paragraphb{Page fault-driven access to remote memory} When a user application tries to access a memory address not present in the compute blade cache, a page fault handler is triggered and the compute blade kernel sends a one-sided RDMA read request to the switch with the virtual address and the requested permission class, i.e., read or write for Linux. At the same time, the page to be used by the user application is registered to the NIC as the receiving buffer, obviating the need for additional data copies. Once the page is received, the local memory structures such as PTEs are populated and the control is returned to the user. Our implementation of the compute blade DRAM cache is similar to LegoOS~\cite{legoos}, but additionally handles cache invalidations for coherence. Specifically, the cache tracks the set of writable pages locally, and on receiving an invalidation request for a region, it flushes all writable pages in the region and removes all local PTEs.

While the above approach enables transparency for access to disaggregated memory, it presents a significant limitation in our implementation --- it restricts the memory consistency model in \name to stronger Total Store Order (TSO), and precludes weaker consistency models, \eg, Process Store Order (PSO) used in DSM approaches~\cite{gam}. This is because unlike TSO, PSO enables multiple writes to a cached memory region to be propagated asynchronously, but blocks if there is a subsequent read to the same region. Realizing this relaxation using page faults requires such writes to be buffered at the compute blade's local DRAM cache without triggering a page fault, but triggering one on a subsequent read to the same page. Unfortunately, this is impossible in traditional x86 or ARM architectures, since they do not support triggering a trap on read without also triggering one for a write. Consequently, \name's stricter TSO model results in limited scalability for workloads with high read/write contention to shared memory regions, as we show in \S\ref{subsec:macro_bench}. We discuss possible architectural changes to address this in \S\ref{sec:discussion}.

\subsection{Memory blade} 

Unlike prior disaggregated memory systems~\cite{legoos, infiniswap} or distributed shared memory systems~\cite{gam}, \name does not require any compute logic or data plane processing logic to run on the memory blades, obviating the need for general purpose CPUs on them. However, since the memory blade in our prototype is realized on traditional Linux servers, we rely on the kernel module at the memory blade to perform RDMA-specific initializations. When a memory blade comes online, its kernel module registers physical memory addresses to the RDMA NIC and reports the mapped address to the global controller. However, subsequent one-sided RDMA requests from compute blades are handled completely by the memory blade NIC without involving the CPUs. Ideally, memory blades would be realized with all logic, including initialization, completely in hardware, without a CPU. While this would facilitate a memory blade design that is both simple and cheap, it would require new hardware design.

\subsection{Programmable Switch}
\label{ssec:switchimpl}

The \name programmable switch module is implemented on a 32-port EdgeCore Wedge switch with a $6.4$~Tbps Tofino ASIC and an Intel Broadwell processor, $8$~GB RAM and $128$~GB SSD. The general purpose CPU hosts the \name control program, which performs process, memory and cache directory management. Meanwhile, the ASIC performs address translation (\S\ref{subsec:addr_trans}) and memory protection (\S\ref{subsec:mem_prot}), handles directory state transitions and virtualizes RDMA connections between compute and memory blades. We here provide implementation details of the mechanisms not already described in \S\ref{sec:design}.

\paragraphb{Process \& memory management} The control plane hosts a TCP server to handle system call intercepts from compute blades, and maintains traditional Linux data structures for process/thread management (\code{task\_struct}) as well as memory management (\code{mm\_struct}, \code{vm\_area\_struct}). On receiving a system call, the control plane modifies the data structures accordingly, and responds with return values consistent with the system calls to ensure transparency.

\paragraphb{Cache directory management} \name reserves a fixed amount of SRAM at the data plane for storing directory entries, and partitions it into fixed sized slots, one for each \textit{region} entry in \name. The control plane maintains a \textit{free list} for available slots, and a hash table \textit{used map} which maps the base virtual address for the dynamically sized cache region to the SRAM slot storing its directory entry. All slots are initially added to the free list. \name creates a directory entry for a region during its allocation by removing an SRAM slot from the free list, populating it with the directory entry with invalid (\textbf{I}) state, creating a match-action rule that maps the block virtual address to the SRAM slot at the data plane, and updating its \textit{used map}. A similar process occurs when a region is split, while removing a directory entry follows the reverse procedure.

\begin{figure}[t]
  \centering
  \includegraphics[width=0.97\columnwidth]{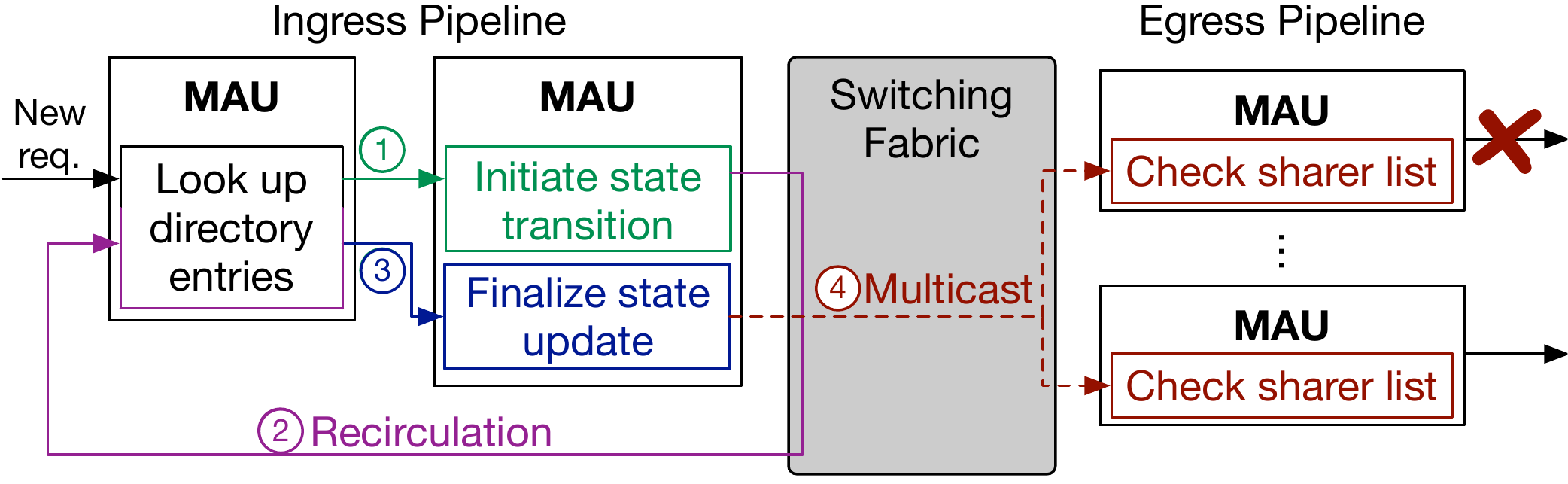}
  \caption{\textbf{Performing directory state transitions on switch ASIC.}}
  \label{fig:cc_example}
\end{figure}

\paragraphb{Directory state transitions} We found that a single match-action unit (MAU) in today's switch ASICs is unable to (i) perform a directory entry lookup, (ii) determine the correct translation based on the current block state and memory access request, and (iii) update the directory entry accordingly all at once, due to their limited compute capability. As such, we split the logic for (i-ii) across two MAUs (\circled{1} in Figure~\ref{fig:cc_example}): the first MAU stores the directory entries and performs (i), while the second MAU stores a \textit{materialized} state-transition table containing all possible transitions and corresponding actions to be performed for (ii). Note that explicitly storing the state-transition table trades-off data plane memory capacity to overcome the limited compute cycles in an MAU. To perform (iii), the second MAU \textit{recirculates} the memory access request packet within the switch data plane to send it back to the first MAU (\circled{2}), so that it can update the directory entry according to actions determined by the second MAU (\circled{3}). If the state transition requires cache invalidations, the data plane creates invalidation requests leveraging \textit{multicast}, as described in \S\ref{subsec:cache_dir}. Specifically, these requests are only forwarded to the current sharers for the relevant page (\circled{4}).

\paragraphb{Virtualizing RDMA connections} When a compute blade in \name issues an RDMA request, it does not know the location of the blade where the page resides. Consequently, the switch data plane in \name \textit{virtualizes} RDMA connections between all possible compute-memory blade pairs by transparently transforming and redirecting corresponding RDMA requests and responses between them. Specifically, once an RDMA request's destination blade is identified via address translation or cache coherence, the data plane updates the request's packet header fields such as IP/MAC addresses and RDMA-specific parameters, before forwarding it to the blade.

% !TEX root = ../paper.tex
\section{Evaluation}
\label{sec:evaluation}

\begin{figure*}[ht!]
  \centering
  \includegraphics[width=0.345\textwidth]{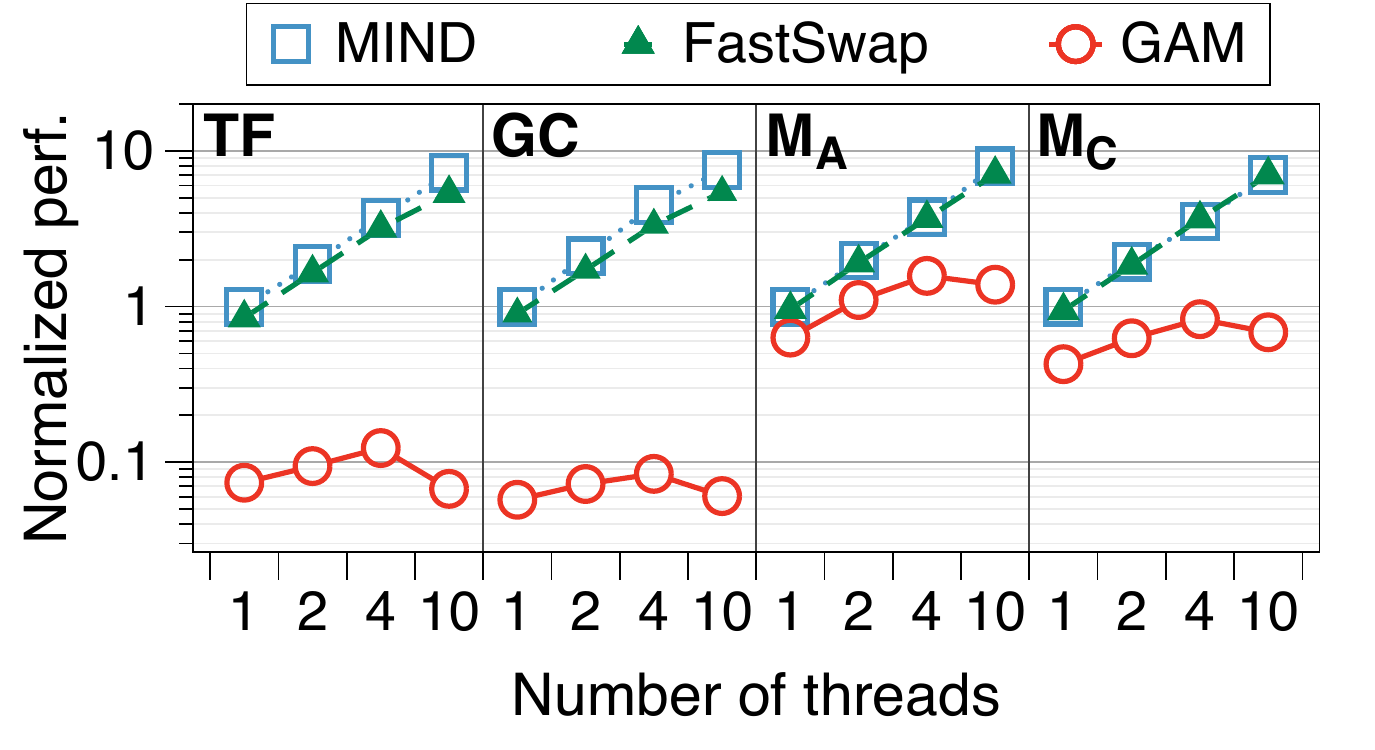}\hspace{-0.25em}
  \includegraphics[width=0.3266\textwidth]{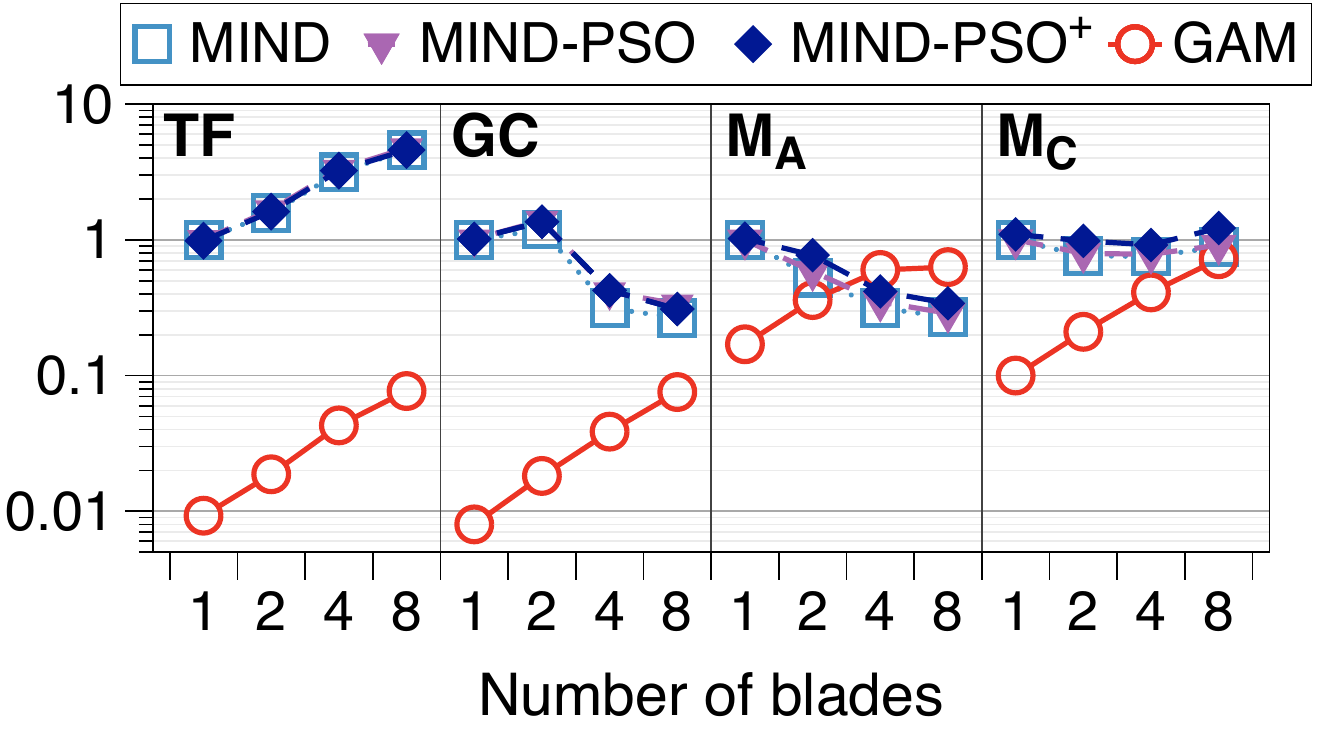}\hspace{-0.25em}
  \includegraphics[width=0.3247\textwidth]{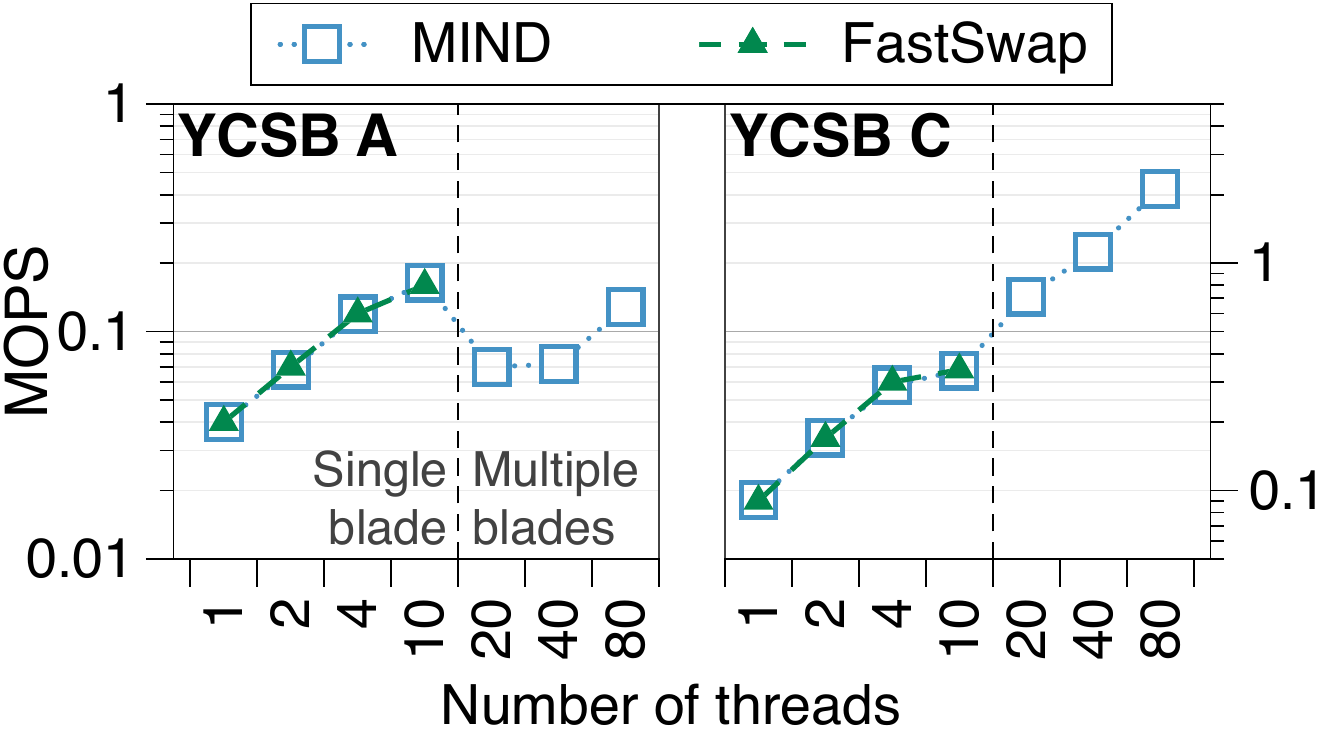}
  \caption{\textbf{Performance scaling} (left) on a single compute blade, (center) across compute blades, and (right) for Native-KVS. For (center), each blade runs 10 threads. Performance is normalized by \name's performance at 1 thread for (left) and 1 blade for (center); the runtimes in seconds for TF, GC, M$_A$ and M$_C$ workloads are 62.8, 59.3, 301.4 and 268.2 for 1 thread, and 69.2, 62.8, 302.3 and 306.9 for 1 blade, respectively.}
\label{fig:perf}
\label{fig:perf_intra}
\label{fig:perf_inter}
\label{fig:perf_kvs}
\end{figure*}

\begin{figure}[h!]
  \centering
  \includegraphics[width=0.45\textwidth]{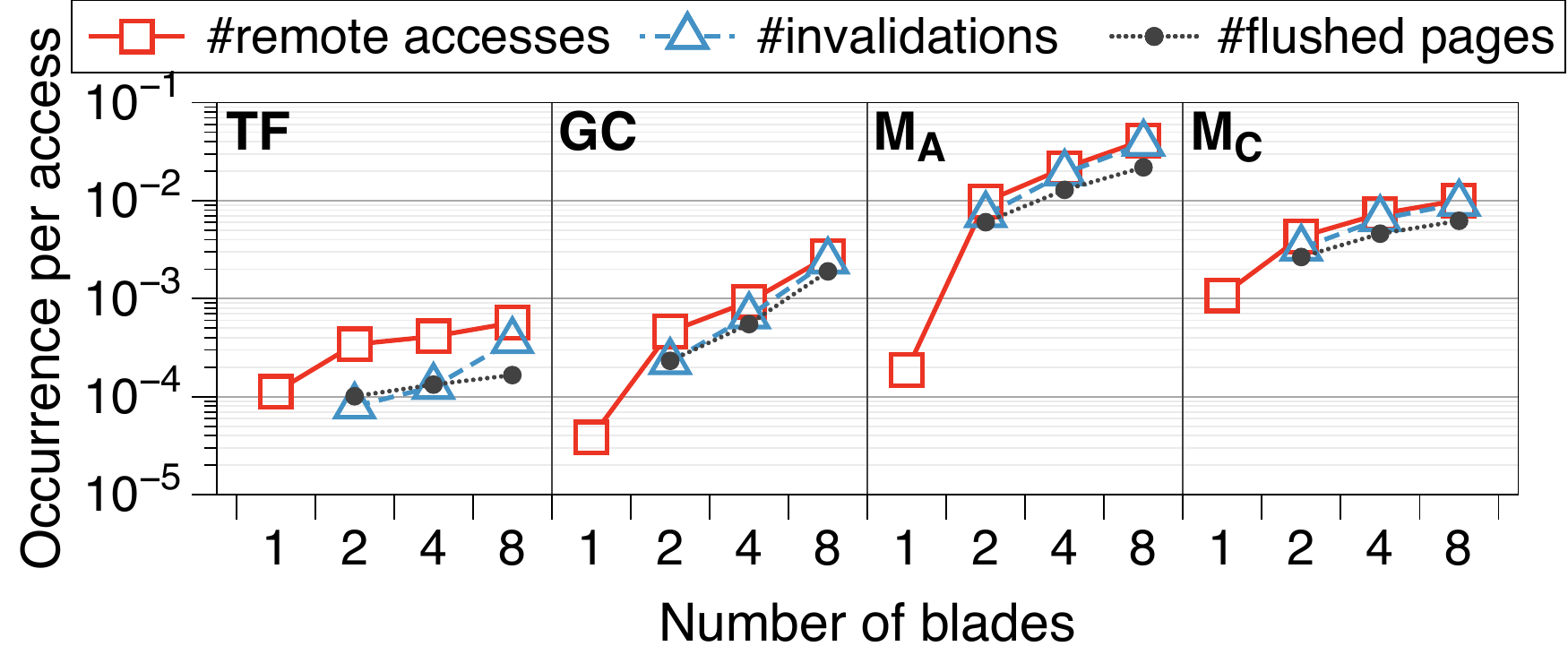}
  \caption{\textbf{Invalidation overhead.}}
\label{fig:perf_intra_inv}
\end{figure}

We evaluate \name to answer the following questions:
\begin{itemize}[topsep=2pt, partopsep=0pt, itemsep=0pt, leftmargin=*]
  \item Can \name enable transparent elasticity over performant disaggregated memory for real-world workloads (\S\ref{subsec:macro_bench})?
  \item What are \name's performance/resource bottlenecks (\S\ref{ssec:bottlenecks})?
  \item How does \algo perform in isolation (\S\ref{ssec:sensitivity})?
\end{itemize}
\paragraphb{Compared systems} We compare \name against two extreme designs in disaggregated memory (\S\ref{ssec:challenges}): a \textit{transparent} DSM-based approach with cache-coherence that supports compute elasticity, and a \textit{non-transparent} approach that limits compute elasticity to a single compute blade. For the former, we adapt GAM~\cite{gam}, a software-based DSM, to our disaggregated setting where cache directory is implemented at the compute blades. For the latter, we employ FastSwap~\cite{fastswap}, a state-of-the-art swap-based disaggregated memory system. All systems employ RDMA for efficient remote memory accesses. Finally, \name uses an initial region size of $16$~kB and epoch size of $100$~ms for its \algo algorithm.

\paragraphb{Cluster setup} We used a cluster comprising five servers connected via the programmable switch described in \S\ref{ssec:switchimpl}. We used a single server equipped with two 18-core Intel Xeon processors, $384$~GB of memory and four Nvidia/Mellanox CX-5 $100$~Gbps NICs, to host multiple memory blade VMs. To highlight the overhead and scalability of in-network cache coherence protocol, we used the remaining four machines, each equipped two 12-core Intel Xeon processors and two Nvidia/Mellanox CX-5 $100$~Gbps NICs, to host two compute blade VMs on each server. Similar to prior work~\cite{legoos}, we emulate the partial disaggregation model by limiting the local DRAM usage at each compute blade to $512$~MB, which is about $25\%$ of the memory footprint for our evaluated workloads. Note that each compute and memory blade VM in our setup had dedicated access to a separate $100$~Gbps NIC to ensure they represent separate network attached entities.

\paragraphb{Applications and workloads} We use several real-world workloads in our evaluation: TensorFlow~\cite{tensorflow} with ResNet-50~\cite{resnet} on CIFAR-10~\cite{cifar10} (denoted as TF), GraphChi~\cite{graphchi} with PageRank~\cite{pagerank} on Twitter graph~\cite{twitter_graph} (denoted as GC), and Memcached~\cite{memcached} with YCSB~\cite{ycsb_workload} workload A ($50\%$ read, $50\%$ write split, denoted as M$_A$) and workload C ($100\%$ reads, denoted as M$_C$). Since GAM is a \textit{software} DSM system, it requires applications to use a specialized memory API, while \name and FastSwap are transparent to applications. To ensure consistent comparison under different interfaces, we captured the memory accesses from our workloads using Intel's PIN~\cite{intel_pin}, and used a memory access emulator to generate the exactly same memory accesses across all three systems. In addition, we also present results for native execution of a simple key-value store (denoted as Native-KVS) on \name and FastSwap, since they support a transparent memory interface.

\begin{figure*}[ht!]
  \centering
  \includegraphics[width=0.305\linewidth]{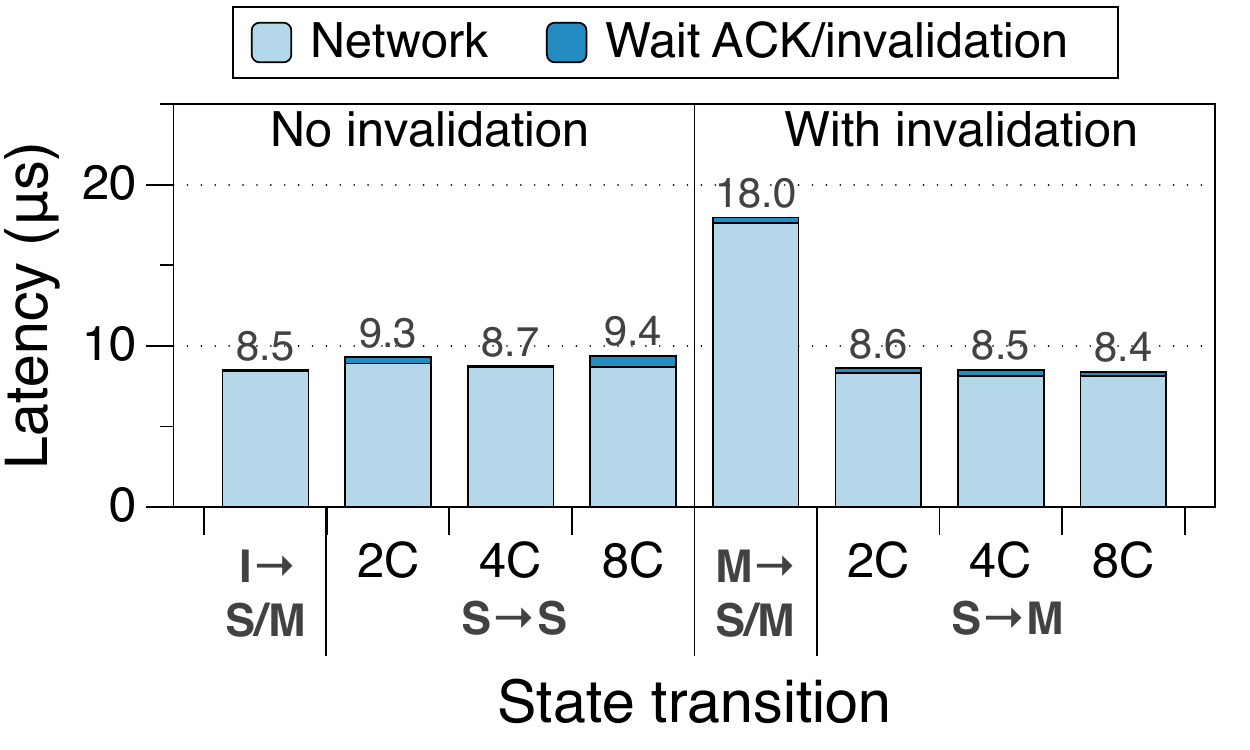}
  \includegraphics[width=0.214\linewidth]{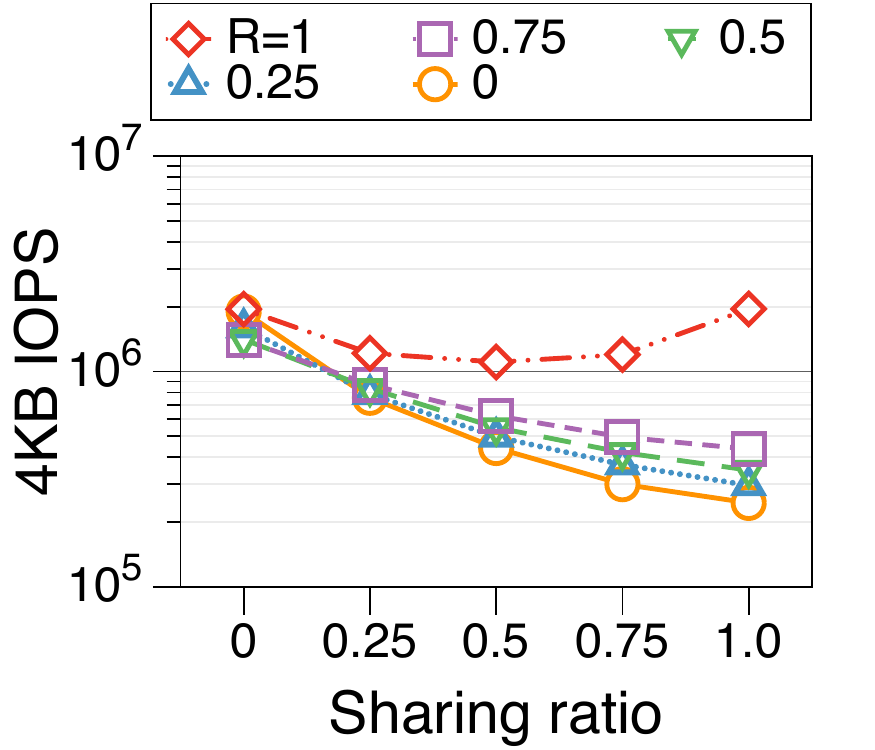}
  \includegraphics[width=0.427\linewidth]{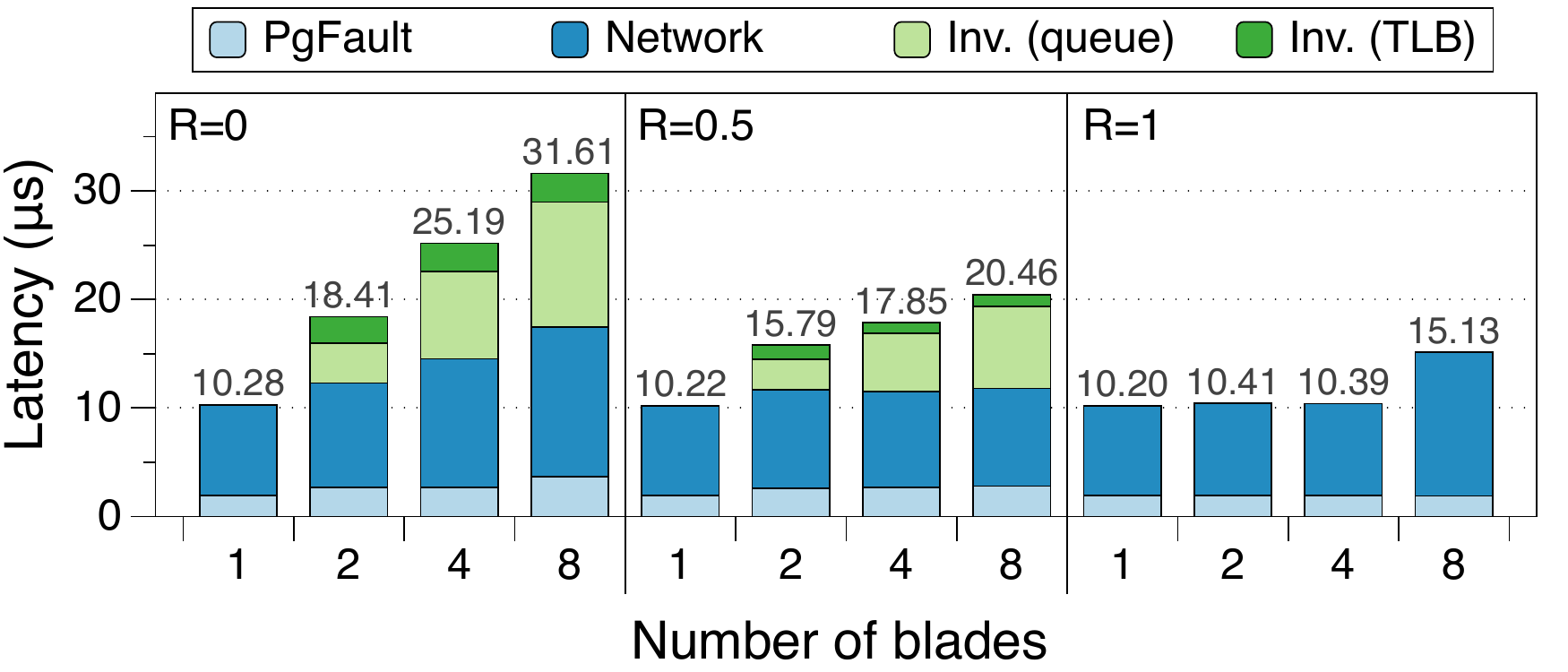}
  \caption{\textbf{Performance bottlenecks.} (left) Network latency for state transitions, (center) memory throughput vs. read-write/sharing ratios, \\(right) network latency break-down}
  \label{fig:micro_latency}
  \label{fig:micro_band}
  \label{fig:micro_latency_load}
  \label{fig:perfbottlenecks}
\end{figure*}

\subsection{Performance Scaling for Real-World Workloads}\label{subsec:macro_bench}

We start by evaluating \name's performance scalability.

\paragraphb{Intra-blade scaling} Figure~\ref{fig:perf_intra}~(left) shows performance scaling across all systems as the number of threads is increased on a single compute blade. We report performance as the inverse of runtime, normalized by the performance of \name for 1 thread. \name and FastSwap scale almost linearly with more compute blades because of their efficient page-fault driven remote memory accesses. In contrast, GAM scales linearly only up to 4 threads, and sub-linearly after that due to software overheads from its user-level library. For instance, GAM must check access permissions for every memory access by acquiring a lock, while \name and FastSwap can leverage the hardware MMU to facilitate the same. Such overheads become significant as the compute resources on a single compute blade come under contention at $10$ threads running on a $12$-core node.

\paragraphb{Inter-blade scaling} Next, we evaluate inter-blade scalability by running $10$ execution threads per-blade, for up to $8$ compute blades. Figure~\ref{fig:perf_inter}~(center) shows our results; here, while \name's default memory consistency model is strict (TSO, \S\ref{sec:impl}), \name-PSO denotes the \textit{simulated} performance of \name with the weaker PSO model (same as GAM). \name-PSO$+$ additionally simulates the effect of infinite switch capacity for directory storage. Since we are forced to simulate \name-PSO and \name-PSO$+$ using traces collected on a real TSO-based system, the traces retain additional TSO-associated queuing delays that we cannot elide, \ie, while our simulations can reorder writes and non-conflicting reads, queuing delays remain; in other words, our \name-PSO and \name-PSO$+$ results are underestimates to potential performance of a hardware-based solution. We omit FastSwap as it does not transparently scale beyond a single compute blade, similar to other disaggregation proposals~\cite{legoos, infiniswap}. Finally, to better understand \name performance, we also measured the number of remote accesses, invalidation requests, and flushed pages (\ie, pages pushed back to the memory blade due to the invalidation) in \name as a fraction of the total number of memory accesses (Figure~\ref{fig:perf_intra_inv}).

For a machine learning workload (TF), \name's performance scales well despite its stricter memory consistency model compared to GAM --- doubling the number of compute blades improves \name's performance by ${\sim} 1.67\times$, with a $59 \times$ speeded compared to GAM at $8$ compute blades. For GC, \name's performance increases from $1$ to $2$ compute blades, but starts to decrease beyond that. This is because GC's graph traversals incur random and often contentious access to shared data compared to machine learning workloads in TF. GC writes ${\sim} 2.5\times$ more data in shared pages than TF, generating significantly more state transitions to modified (\textbf{M}) state, and incurring frequent invalidations (\S\ref{ssec:bottlenecks}). Figure~\ref{fig:perf_intra_inv} additionally shows that the growth in the number of invalidations and flushed pages in GC is much steeper than TF, attributing to the reduced performance at higher parallelism. PSO partly alleviates this overhead by permitting writes to be performed asynchronously, but still does not permit linear scaling beyond $2$ compute blades. Instead, GAM scales better because the performance differential between its local and remote accesses is small --- local accesses are $10\times$ slower than that of MIND (due to software implementation of local accesses), while remote access latencies are similar for both. Consequently, performing more remote accesses (during invalidations) does not impact GAM performance as much as it does for \name.

Finally, M$_A$ and M$_C$ have more sharers with much larger number of shared writes compared to TF and GC. As a result, \name does not scale well beyond $1$ compute blade because: (1) more blades contend for acquiring write permission to the same region incurring multiple invalidations and significantly smaller number of local memory accesses, and (2) the directory storage at the switch becomes a bottleneck (as we show in \S\ref{ssec:bottlenecks}), frequently resulting in false invalidations for heavily shared memory regions. As shown in Figure~\ref{fig:perf_intra_inv}, M$_A$ and M$_C$ trigger significantly more (over $10\times$) invalidations and page flushes than either GC or TF workloads. We also confirm these insights through \name-PSO and \name-PSO$+$ simulated results, which show that employing weaker memory consistency models and infinite directory capacity improves \name's performance to some extent. Note that for M$_C$, \name's performance increases from $4$ to $8$ blades since the number of invalidations do not increase by much. GAM scales better due to its weaker consistency model, and by leveraging its software implementation to facilitate several memory access reorderings which are not possible in \name. Consequently, at $8$ compute blades, GAM and \name-PSO$+$ achieve roughly similar performance.

\paragraphb{Native KVS} Figure~\ref{fig:perf_kvs}~(right) shows the intra- and inter-blade scaling of Native-KVS on \name and FastSwap for YCSB-A and C workloads. On a single blade, both \name and FastSwap observe near linear performance scaling for up to $10$ threads. Since FastSwap does not support sharing state across multiple compute blades, we do not report its performance beyond $10$ threads. Similar to our results for M$_A$, \name does not scale well beyond a single compute blade for the YCSB-A workload ($50\%$ reads, $50\%$ writes) due to high read-write contentions. For the YCSB-C workload, Native-KVS scales linearly even beyond a single blade since it is a read-only workload, incurring no invalidations. Interestingly, YCSB-A workload on Native-KVS scales better than M$_A$ --- we attribute this to better partitioning of KVS state across compute blades in Native-KVS compared to Memcached.

\subsection{\name Performance and Resource Bottlenecks}
\label{ssec:bottlenecks}

We study \name bottlenecks in terms of (i) memory access performance, and (ii) memory resources at the switch.

\begin{figure*}[ht!]
  \centering
  \includegraphics[width=0.3\textwidth]{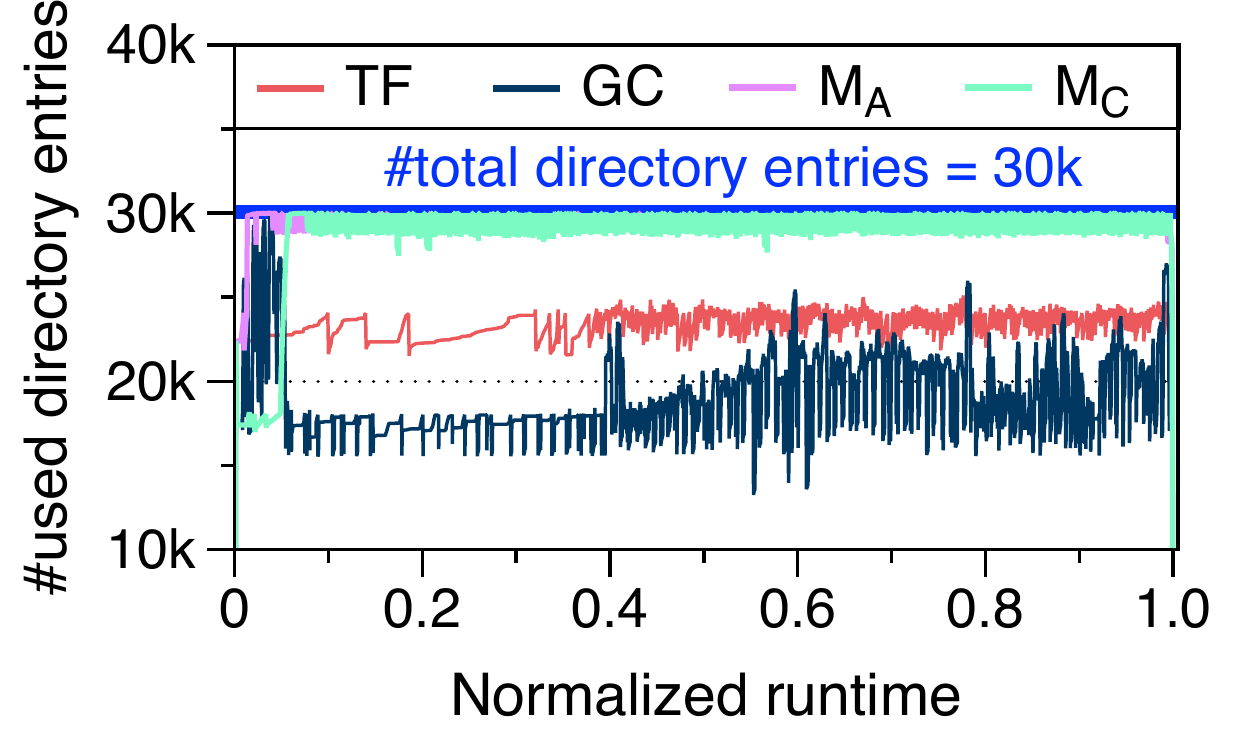}\hspace{1em}
  \includegraphics[width=0.31\textwidth]{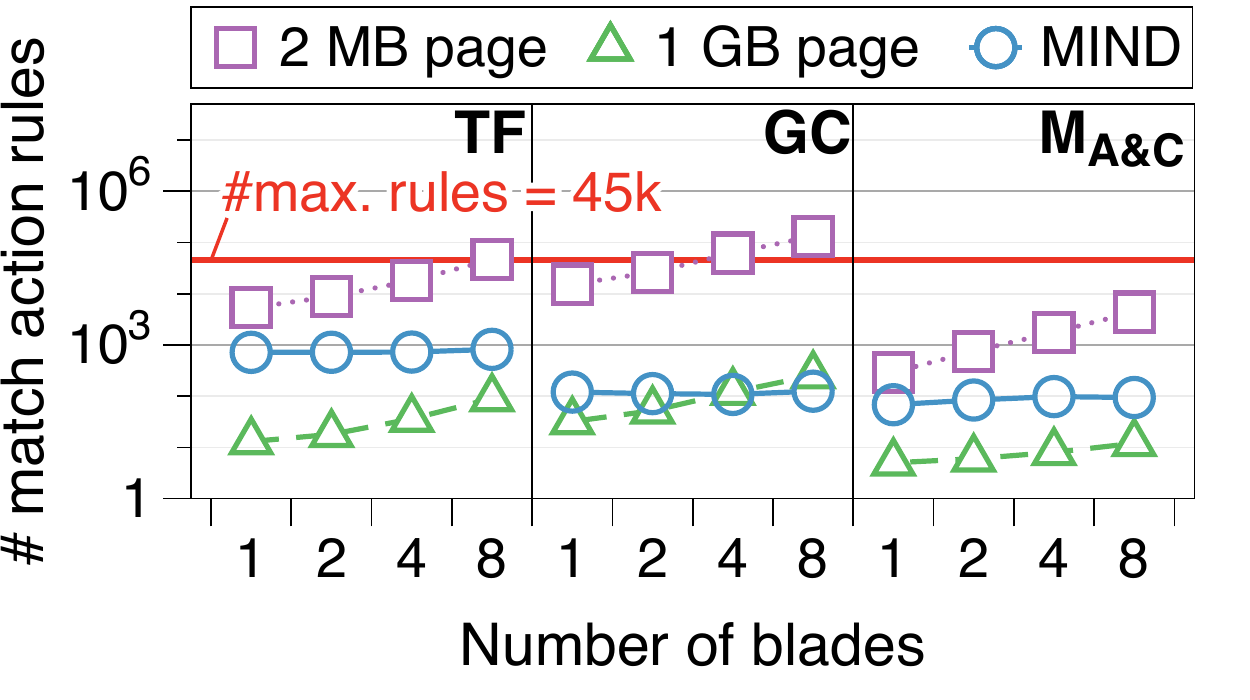}\hspace{1em}
  \includegraphics[width=0.31\textwidth]{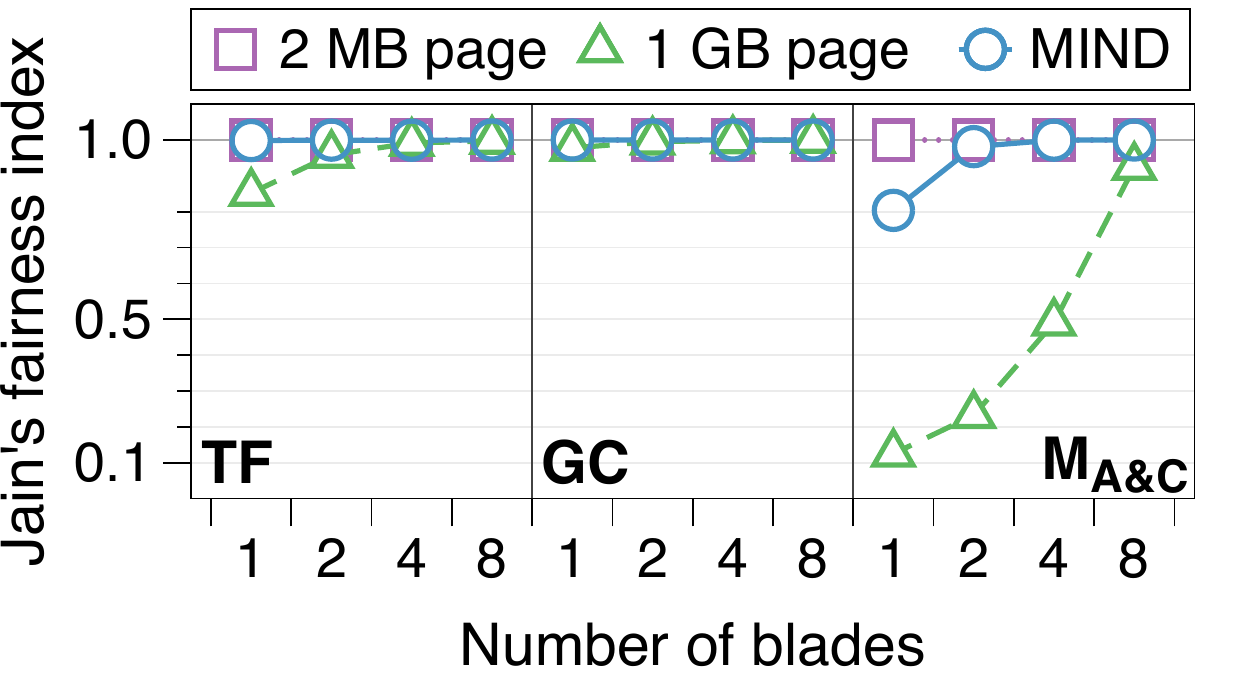}
  \caption{\textbf{\name switch resource bottlenecks.} (left) Directory entries, (center) match-action entries for heap, (right) load balancing for heap.}
  \label{fig:resbottlenecks}
  \label{fig:micro_load_balance}
  \label{fig:micro_num_entry}
  \label{fig:macro_cache_dir}
\end{figure*}

\begin{figure*}[t!]
  \centering
  \includegraphics[width=0.49\textwidth]{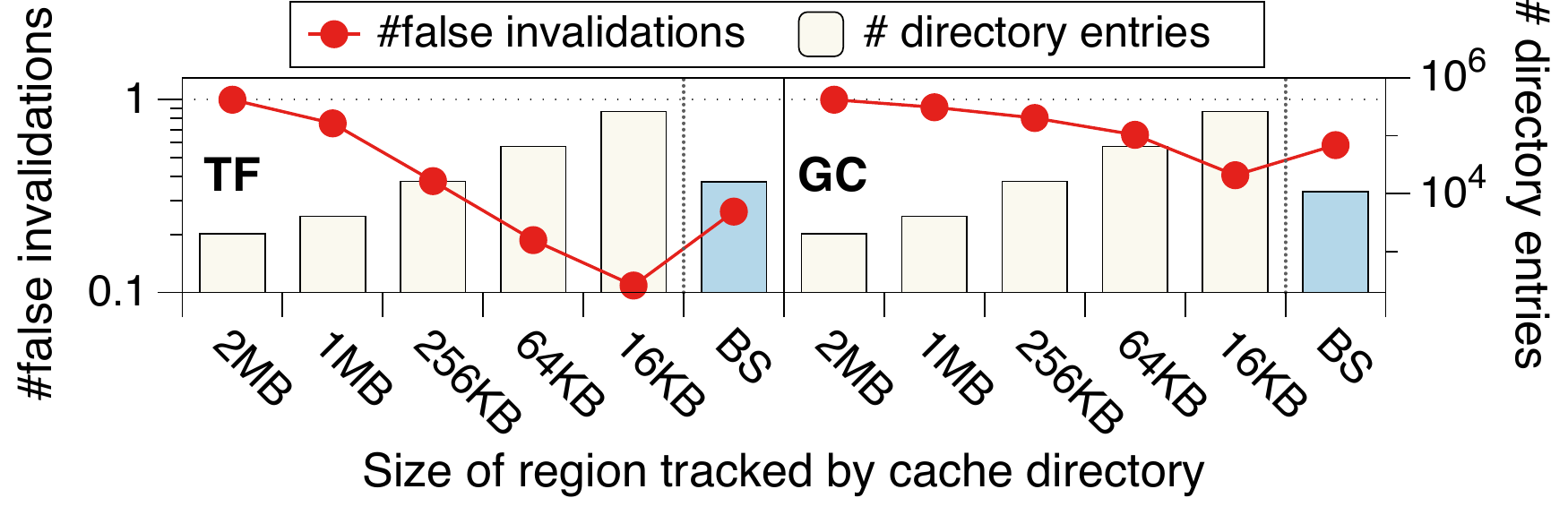}\hspace{2em}
  \includegraphics[width=0.46\textwidth]{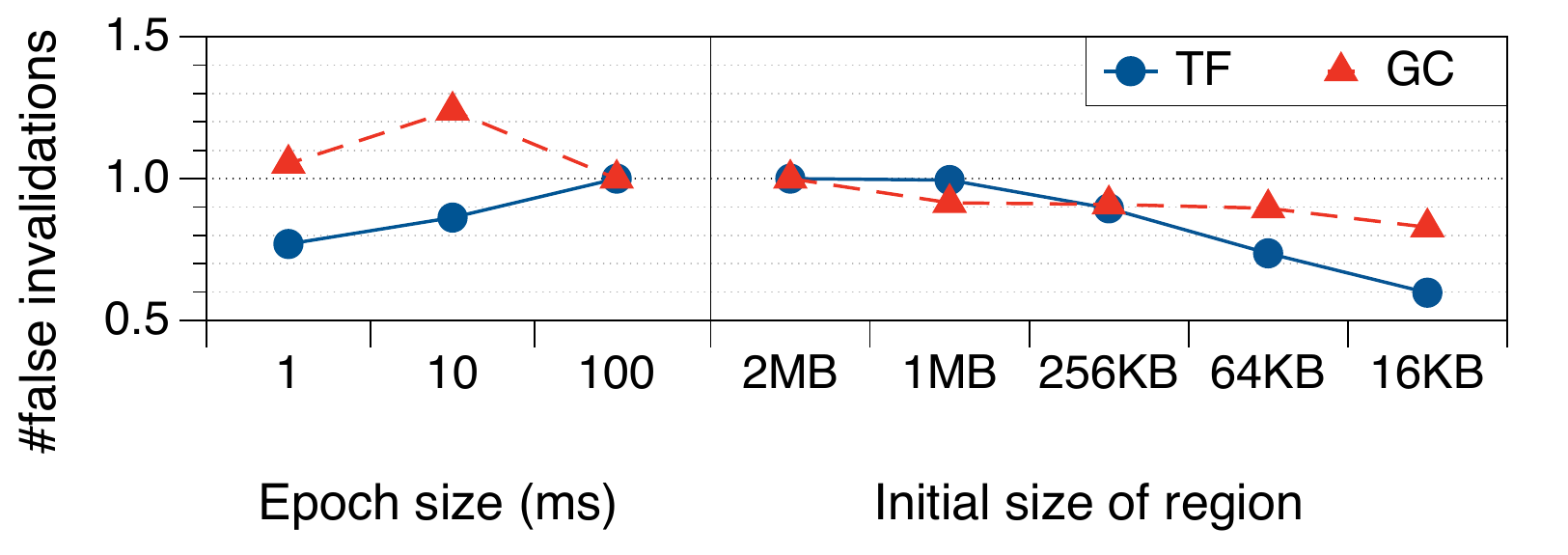}
  \caption{\textbf{Evaluating \name's \algo algorithm.} (left) Navigating switch storage vs. performance tradeoff. (right) Impact of epoch \& initial region sizing. The number of false invalidations is normalized by value at $2$~MB for region size and $100$~ms for epoch size.}
  \label{fig:storagevperf}
  \label{fig:epoch_size}
  \label{fig:region_size}
\end{figure*}%

\paragraphb{Latency for cache state transitions} Figure~\ref{fig:micro_latency} (left) shows the end-to-end latency due to every possible state transition under the MSI protocol in \name, including the time required to fetch the data. Note that this figure only shows latency for remote accesses --- local accesses only incur DRAM latency ($<100$~ns). On the x-axis, $2$ -- $8$C indicate the number of compute blades requesting the same page, and $x \rightarrow y$ denotes the state transition, $x$ and $y$ being the initial and final states. 

When a blade requests read-only (shared, \textbf{S}) mode for a region, and its initial state was either invalid (\textbf{I}) or shared (\textbf{S}), it does not require any invalidations. Consequently, the data fetch can be performed in a single RDMA request (${\sim} 9~\mu$s), as seen in the first four bars. If the transition for a region is either from or to the modified (\textbf{M}) state, the requesting blade must wait until the regions is invalidated at all its previous owners. When transitioning from \textbf{S} to \textbf{M}, the data can be fetched directly from the memory blade via one-sided RDMA operation, while the invalidation at other blades occur in parallel, resulting in a total latency of ${\sim} 9~\mu$s. When the region is initially in \textbf{M} state, the (dirty) data must be fetched from and the region invalidated at the same blade --- its current owner. Therefore, the invalidation and data fetch occur sequentially, resulting in ${\sim} 18~\mu$s latency. Note that since the latency for requests with invalidations is $2\times$ higher than requests without them, a workload's performance depends on the relative proportion of the different types of requests, as we show next.

\paragraphb{Impact of invalidations on memory throughput} Figure~\ref{fig:micro_band}~(center) shows \name memory throughput across 8 compute blades, running 1 compute thread each, under various read-write and sharing characteristics. We use read ratio to denote the fraction of reads in the workload (remaining accesses are writes), and sharing ratio to denote the portion of memory accesses that occur to a shared region (shared by all threads). We used a total working set size of $400~k$ pages, with the access pattern across them being uniform random. If most accesses are reads, then compute blades can share the same region without triggering invalidations (\textbf{S}$\rightarrow$\textbf{S} in Figure~\ref{fig:micro_latency} (left)). As such, at read-ratio $1$, most of the pages are accessed locally from the cache, resulting in very high memory throughput ($1$-$2\times 10^{6}$ IOPS) for all sharing ratios. Again, at sharing ratio $0$, memory throughput remains high, since accessed pages can remain cached at the compute blade without being invalidated, \ie, most accesses are local. If both the write proportion and sharing ratio are increased, memory throughput drops (by ${\sim} 10\times$ at sharing-ratio $1$), since they trigger a large number number of \textbf{M}$\rightarrow$\textbf{S}, \textbf{S}$\rightarrow$\textbf{M} transitions with invalidations and permit few pages to be accessed locally.

\paragraphb{Impact of invalidations on memory latency} Figure~\ref{fig:micro_latency_load} (right) shows the end-to-end latency for the same setup as Figure~\ref{fig:micro_band}~(center). The sharing ratio is fixed at $1$ (\ie, all pages are shared), while the read-ratio and number of compute blades are varied. The read-only workload (at read-ratio $\texttt{R}=1$) observes end-to-end latency close to \textbf{S}$\rightarrow$\textbf{S} transitions in Figure~\ref{fig:micro_latency}~(left) regardless of number of compute blades. However, lower read-ratios ($0.5$ and $0$) experience higher end-to-end latency compared to \textbf{M}/\textbf{S}$\rightarrow$\textbf{M} transitions in Figure~\ref{fig:micro_latency}~(left), due to two additional overhead sources: TLB shootdowns and queueing delays. First, each invalidation request may cause a page table entry (PTE) at the compute blade to be unmapped or experience a permission change (\eg, writeable to read-only), causing a corresponding TLB shootdown. These shootdowns must be performed synchronously for correctness, and can incur several microseconds of latency (Inv. (TLB)), similar to observations in prior work~\cite{latr}. Second, the increased number of invalidations at larger number of compute blades and lower read-ratios additionally leads to longer queueing delay (Inv. (queue)) --- time an invalidation request has to wait to be processed in a compute blade.

\paragraphb{Cache directory storage} Figure~\ref{fig:macro_cache_dir}~(left) shows the number of cache directory entries stored in the switch data plane in \name over time for the workloads evaluated in \S\ref{subsec:macro_bench} across 8 compute blades, running 10 threads each. In \name, we fix the total amount of storage allocated to directory storage to $30$~k entries. For the TF and GC workloads, \name's \algo algorithm ensures that the number of directory entries remains well below the limit over time. However, the MC$_{A}$ and MC$_{B}$ workloads have a significantly larger number of shared memory regions, with frequent read and write accesses to them; as a consequence, the number of directory entries for the workloads always remains close to the $30$~k limit. Recall from \S\ref{subsec:macro_bench} that one of the key reasons for poor scalability of these workloads is the number of false invalidations triggered due to the relatively coarse granularity of tracking directory entries --- we believe with future switch ASICs likely to be equipped with more TCAM/SRAM, this bottleneck would no longer exist, permitting more efficient scaling under \name.

\paragraphb{Address translation \& memory protection storage} We study the switch storage overheads due to address translation and memory protection on a setup with 8 memory blades, running the TF, GC, and M$_{A/C}$ workloads; we group $M_A$ and $M_C$ since they have the same memory allocations. Figure~\ref{fig:micro_num_entry}~(center) shows that the number of match action rules due to address translation and memory protection in \name is almost constant, even as the workload size increases. This is due to \name's per-memory blade partitioning of the address space, and \code{vma} granularity tracking of memory protection entries. While we have only shown results for three different applications, we find that the number of \code{vma} entries for typical datacenter applications falls in similar ranges, and well under $1$--$2$~k~\cite{vma1, vma2}. In contrast, the number of match-action rules increases linearly with the dataset size for page-based approaches, despite smaller absolute overheads with $1$~GB huge pages. Note that the upper-limit for match-action rules that the switch can store is about $45$~k --- higher than the $30$~k limit for directory entries due to a more compact representation.

\name's memory allocation also ensures balanced placement of load across memory blades (\S\ref{subsec:addr_trans}), as shown via Jain's fairness index metric~\cite{jain} in Figure~\ref{fig:micro_load_balance}~(right). While $2$~MB pages can achieve similar load-balancing, they do so at the cost of much larger number of address translation entries. $1$~GB pages, on the other hand, observes poor load balancing for allocation-intensive workloads (M$_{A/C}$).

\subsection{Evaluating \name's \Algo Algorithm}
\label{ssec:sensitivity}

We now evaluate \name's \algo algorithm.

\paragraphb{Storage vs. performance tradeoff} Recall from \S\ref{ssec:caching} that the granularity at which the directory tracks memory regions exposes a tradeoff between the false invalidation count and the size of the directory itself --- Figure~\ref{fig:storagevperf}~(left) highlights this tradeoff for the TF and GC workloads. Specifically, tracking smaller regions (\eg, $16$~kB) permits fewer false invalidations, but at the cost of larger number of directory entries at the switch, while tracking larger regions (\eg, $2$~MB) exposes the opposite tradeoff. \name's \algo algorithm employs adaptive region sizing to balance both the number of directory entries as well false invalidations.

\paragraphb{Impact of epoch and initial region sizing} Figure~\ref{fig:epoch_size}~(right) shows the impact of epoch size on the total number of false invalidations for TF and GC workloads. Increasing the epoch size from $1$ to $100$~ms does not have a significant impact on the number of false invalidations, but reduces the control plane overheads. Epoch size smaller than $1$ms (not shown) are unable to capture enough invalidations to enable accurate estimation of the distribution, resulting spurious merges/splits and unpredictable false invalidations. We use $100$~ms as our default epoch size since it offers a sweet spot for minimizing both false invalidations and control plane overheads.

Figure~\ref{fig:region_size}~(right) shows that picking smaller initial region sizes results in fewer false invalidations --- intuitively, this is because larger initial region sizes require several splits before stabilizing to the appropriate region size, incurring several false invalidations in the interim. We select $16$KB as our default initial region size, since smaller region sizes result in too many directory entries during initialization. 

Finally, we note that neither parameter has any noticeable impact on the number of directory entries at stable state.

% !TEX root = ../paper.tex
\section{Limitations and Future Research}
\label{sec:discussion}

We now discuss the limitations of current \name implementation, and future research directions to resolve them. 

\paragraphb{Thread management} Even with our optimizations, remote memory access latency is still at least two orders of magnitude higher than local latency. While our work explores in-network approaches to minimize overheads of coherence, an orthogonal approach of co-locating threads with higher proportion of shared memory accesses could yield significant improvements in end-to-end application performance by reducing the number of invalidations over the network.

\paragraphb{Other coherence protocols} While \name implements the simple MSI coherence protocol, more complex protocols like MOESI may offer better scalability by reducing broadcasts and write-backs to disaggregated memory. Realizing such protocols would require storing larger state transition tables (STT) at the switch and handling more transient states, adding implementation complexity for ensuring correctness. Still, the number of TCAM entries required for STT entries would be quite small (\eg, tens of states for MOESI) relative to switch ASIC capacities, making them realizable today.

\paragraphb{Weaker consistency models} As we noted in \S\ref{sec:impl} and \S\ref{sec:evaluation}, our page-fault based implementation on x86 architectures cannot realize weaker consistency models like PSO. To this end, a redesign of the compute blade architecture --- \eg, by enabling page-faults on reads (but not writes) to a page --- could enable realization of weaker consistency models in \name, facilitating higher throughput to disaggregated memory.

\paragraphb{Scaling beyond a rack} While \name targets a rack-scale design with a single switch, some workloads may want to scale transparently beyond a single rack. This requires a shift similar to the shift from single node CPUs (akin to the rack in our setting) to multi-node NUMA architectures (akin to datacenter-scale memory disaggregation). Such a design would require extension of \name design from a single switch to a datacenter-wide network topology.

\paragraphb{Virtualization} While \name enables protection at a virtual memory level, extensions to virtualization are needed to facilitate true isolation across users for security, resource management, legacy OS support, \etc~Providing performance isolation, in particular, would require isolating several different shared resources along the compute-memory interconnect, including network bandwidth, switch and NIC resources.

% !TEX root = ../paper.tex
\section{Related Work}
\label{sec:related}

While we discussed prior disaggregated memory approaches in \S\ref{ssec:challenges}, we now discuss other work related to \name. 

\paragraphb{In-network computing} There have been several recent efforts that leverage in-network computing for performance gains~\cite{innetwork, pktsched, congestion1, congestion2, lb1, lb2, netpaxos, p4xos, nopaxos, synchrony, concurrency1, concurrency2, netlock, agg1, agg2, agg3, qp1, qp2}. Most focus was on offloading \textit{application} logic and state to the network, \eg, key-value caches~\cite{netcache, incbricks} and metadata~\cite{pegasus}. Perhaps the most relevant to \name are NOPaxos~\cite{nopaxos} and Concordia~\cite{concordia}. NOPaxos leverages the network to order requests for Paxos-based consensus, enabling consistent replication without expensive coordination overheads. While \name targets a complementary goal of in-network memory management, it could leverage NOPaxos to enable consistent replication of disaggregated memory. Concordia, on the other hand, uses the programmable switch as a cache for directory entries in a DSM; in contrast, \name realizes memory management completely in the network.

\paragraphb{Application-driven memory disaggregation} Recent work argues for OSes to expose resource management abstractions like memory placement and failures to the applications for high performance and better fault-tolerance~\cite{disaggapp, disaggfault}. While \name argues for a transparent disaggregated shared memory, it is not incompatible with the above approaches --- OS-level libraries layered atop \name could still expose memory placement and failure notifications to applications. 

Clover~\cite{clover} explores the design of a key-value store closely integrated with disaggregated persistent memory. In contrast to \namex's in-network transparent memory management, Clover focuses on lock-free consistent access to KV pairs stored in network-attached memory using atomic RDMA verbs. While a possible design considered in~\cite{clover} places key-value coordination and access logic at a centralized coordinator, it does not place the logic in the network fabric and does not consider memory protection, caching or coherence.

\paragraphb{Emerging industry standards} While most industry standards for high performance compute-memory interconnects like CCIX~\cite{ccix}, CXL~\cite{cxl} and OpenCAPI~\cite{opencapi} target \textit{intra}-server settings, Gen-Z~\cite{genz} is perhaps the closest to \name since it targets \textit{inter}-server fabrics. The Gen-Z standard defines operations like \code{ExclusiveRead} and \code{Writeback} that may be used as building blocks for software-based coherence~\cite{genz, genz1}, although we are unaware of any publicly available realization. Moreover, the MMU functionalities in all the above industry standards are realized at the \textit{endpoints}, \eg, at specialized ZMMUs at CPU and memory nodes in Gen-Z; the \textit{fabric} (\eg, the switch) only forwards memory requests and responses. This is in contrast to \name's approach of in-network memory management, \ie, \name's design is complementary to the industry efforts towards high-performance interconnects.

% !TEX root = ../paper.tex
\section{Conclusion}

We have presented \name, an in-network memory management unit for rack-scale memory disaggregation. \name achieves resource elasticity, performance and transparency through a principled redesign of traditional memory management mechanisms to achieve their individual goals in the disaggregated setting while operating under programmable switch ASIC resource constraints. Our \name prototype facilitates transparent resource elasticity, while matching the performance of prior memory disaggregation proposals for real-world workloads.
% !TEX root = ../paper.tex
\section*{Acknowledgements}
We would like to thank our shepherd Yiying Zhang and anonymous SOSP reviewers for their valuable comments and insightful feedback. We are also grateful to Daehyeok Kim, Victor Gomez and Jonathan Kraft for inputs at various stages of the work. This work is supported in part by NSF Awards \code{2047220}, \code{2016422} and \code{1916817} and their REU supplements.

{
% \balance
\bibliographystyle{unsrt}
\bibliography{bib/abr-short,bib/paper_ref}
}
\end{document}